\documentclass{article}
\usepackage[left=1in, right=1in, top=1in, bottom=1in]{geometry}
\usepackage{enumitem}
\usepackage{authblk}
\usepackage{graphicx}	
\usepackage{amsmath, amsthm, amssymb}
\usepackage{xspace}
\usepackage{comment}
\usepackage{hyperref}
\usepackage[dvipsnames]{xcolor}
\usepackage{algorithm}
\usepackage{algpseudocode}
\usepackage{todonotes}
\usepackage{mathtools}
\usepackage[capitalize]{cleveref}
\usepackage{ifthen}
\usepackage[mathlines]{lineno}

\newcommand{\algprobm}[1]{\textsc{#1}\xspace}

\newcommand{\betacc}[1]{\ifthenelse{\equal{#1}{1}}{\exists^{\log n}}{\exists^{\log^{#1}n}}} % {\beta_{#1}}
\newcommand{\alphacc}[1]{\ifthenelse{\equal{#1}{1}}{\forall^{\log n}}{\forall^{\log^{#1}n}}} % {\alpha_{#1}}

\newcommand{\Soc}{\text{Soc}}

\renewcommand{\setminus}{\mysetminus}
\newcommand{\mysetminusD}{\raisebox{.8pt}{\hbox{\tikz{\draw[line width=0.6pt,line cap=round] (3.5pt,0pt) -- (0,5.2pt);}}}}
\newcommand{\mysetminusT}{\mysetminusD}
\newcommand{\mysetminusS}{\raisebox{.5pt}{\hbox{\tikz{\draw[line width=0.45pt,line cap=round] (2.2pt,0) -- (0,3.8pt);}}}}
\newcommand{\mysetminusSS}{\raisebox{.35pt}{\hbox{\tikz{\draw[line width=0.4pt,line cap=round] (1.5pt,0) -- (0,2.8pt);}}}}

\newcommand{\mysetminus}{\mathbin{\mathchoice{\mysetminusD}{\mysetminusT}{\mysetminusS}{\mysetminusSS}}}

\theoremstyle{plain}
\newtheorem{theorem}{Theorem}[section]
\newtheorem{proposition}[theorem]{Proposition}
\newtheorem{corollary}[theorem]{Corollary}
\newtheorem{lemma}[theorem]{Lemma}

\newtheorem{observation}[theorem]{Observation}

\theoremstyle{definition}
\newtheorem{definition}[theorem]{Definition}
\newtheorem{remark}[theorem]{Remark}
\newtheorem{question}[theorem]{Question}

\DeclareMathOperator{\Aut}{Aut}

\DeclareMathOperator{\rad}{Rad}

\DeclareMathOperator{\poly}{poly}

\newcommand*{\ComplexityClass}[1]{\ensuremath{\mathsf{#1}}\xspace}

\newcommand{\FOPLL}{\ensuremath{\textsf{FOLL}^{O(1)}}}

\newcommand*{\LogSpace}{\ComplexityClass{L}}

\newcommand{\AC}{\ComplexityClass{AC}}
\newcommand{\SAC}{\ComplexityClass{SAC}}

\newcommand{\ACz}{\ComplexityClass{AC^0}}

\newcommand{\qACz}{\ComplexityClass{quasiAC^0}}
\newcommand{\FOLL}{\ComplexityClass{FOLL}}

\newcommand{\DLOGTIME}{\ComplexityClass{DLOGTIME}} %
\title{Parallel Algorithms for Group Isomorphism via Code Equivalence}

\author[1]{Michael Levet}
\affil[1]{Department of Computer Science, College of Charleston}

\begin{document}
\maketitle
%\linenumbers
\begin{abstract}
In this paper, we exhibit $\textsf{AC}^{3}$ isomorphism tests for coprime extensions $H \ltimes N$ where $H$ is elementary Abelian and $N$ is Abelian; and groups where $\text{Rad}(G) = Z(G)$ is elementary Abelian and $G = \text{Soc}^{*}(G)$. The fact that isomorphism testing for these families is in $\textsf{P}$ was established respectively by Qiao, Sarma, and Tang (STACS 2011), and Grochow and Qiao (CCC 2014, \textit{SIAM J. Comput.} 2017). 

The polynomial-time isomorphism tests for both of these families crucially leveraged \emph{small} (size $O(\log |G|)$) instances of \algprobm{Linear Code Equivalence} (Babai, SODA 2011). Here, we combine Luks' group-theoretic method for \algprobm{Graph Isomorphism} (FOCS 1980, \textit{J. Comput. Syst. Sci.} 1982) with the fact that $G$ is given by its multiplication table, to implement the corresponding instances of \algprobm{Linear Code Equivalence} in $\textsf{AC}^{3}$.

As a byproduct of our work, we show that isomorphism testing of arbitrary central-radical groups is decidable using $\textsf{AC}$ circuits of depth $O(\log^3 n)$ and size $n^{O(\log \log n)}$. This improves upon the previous bound of $n^{O(\log \log n)}$-time due to Grochow and Qiao (\emph{ibid}.).
\end{abstract}

\thispagestyle{empty}

\newpage

\setcounter{page}{1}

\section{Introduction}
The \algprobm{Group Isomorphism} problem (\algprobm{GpI}) takes as input two finite groups $G$ and $H$, and asks if there exists an isomorphism $\varphi : G \to H$. When the groups are given by their multiplication (Cayley) tables, it is known that $\algprobm{GpI}$ belongs to $\textsf{NP} \cap \textsf{coAM}$. The generator-enumeration strategy has time complexity $n^{\log_{p}(n) + O(1)}$, where $n$ is the order of the group and $p$ is the smallest prime dividing $n$. This strategy was independently discovered by Felsch and Neub\"user \cite{FN} and Tarjan (see \cite{MillerTarjan}). The parallel complexity of the generator-enumeration strategy has been gradually improved \cite{LiptonSnyderZalcstein, Wolf, WagnerThesis, ChattopadhyayToranWagner, TangThesis}. Collins, Grochow, Levet, and Weiß recently improved the bound for the generator enumeration strategy to
\[
\exists^{\log^{2} n} \forall^{\log n} \exists^{\log n}\textsf{DTISP}(\text{polylog}(n), \log(n)),
\]
which can be simulated by depth-$4$ $\textsf{quasiAC}^{0}$ circuits of size $n^{O(\log n)}$ \cite{CGLWISSAC}. Algorithmically, the best known bound for \algprobm{GpI} is $n^{(1/4)\log(n)+O(1)}$, due to Rosenbaum \cite{Rosenbaum2013BidirectionalCD} and Luks \cite{LuksCompositionSeriesIso} (see \cite[Sec. 2.2]{GR16}). Even the impressive body of work on practical algorithms for this problem in more succinct input models, led by Eick, Holt, Leedham-Green and O'Brien (e.\,g., \cite{BEO02, ELGO02, BE99, CH03}) still results in an $n^{\Theta(\log n)}$-time algorithm in the general case (see \cite[Page 2]{WilsonSubgroupProfiles}).

In practice, such as working with computer algebra systems, the Cayley model is highly unrealistic. Instead, the groups are given by generators as permutations or matrices, or as black-boxes. In the setting of permutation groups, \algprobm{GpI} belongs to $\textsf{NP}$ \cite{LuksReduction}. When the groups are given by generating sets of matrices, or as black-boxes, $\algprobm{GpI}$ belongs to $\textsf{Promise}\Sigma_{2}^{p}$ \cite{BabaiSzemeredi}; it remains open as to whether $\algprobm{GpI}$ belongs to $\textsf{NP}$ or $\textsf{coNP}$ in such succinct models.  In the past several years, there have been significant advances on algorithms with worst-case guarantees on the serial runtime for special cases of this problem; we refer to \cite{GQCoho, DietrichWilson, GrochowLevetWL} for a survey.

Key motivation for \algprobm{GpI} comes from its relationship to the \algprobm{Graph Isomorphism} problem (\algprobm{GI}). When the groups are given verbosely by their multiplication tables, \algprobm{GpI} is $\textsf{AC}^{0}$-reducible to $\algprobm{GI}$ \cite{ZKT}. On the other hand, there is no reduction from \algprobm{GI} to \algprobm{GpI} computable by $\textsf{AC}$ circuits of depth $o(\log n / \log \log n)$ and size $n^{\text{polylog}(n)}$ \cite{ChattopadhyayToranWagner}. In light of Babai's breakthrough result that $\algprobm{GI}$ is quasipolynomial-time solvable \cite{BabaiGraphIso}, $\algprobm{GpI}$ in the Cayley model is a key barrier to improving the complexity of $\algprobm{GI}$. There is considerable evidence suggesting that $\algprobm{GI}$ is not $\textsf{NP}$-complete \cite{Schoning, BuhrmanHomer, ETH, BabaiGraphIso, GILowPP, ArvindKurur}. As verbose $\algprobm{GpI}$ reduces to $\algprobm{GI}$, this evidence also suggests that $\algprobm{GpI}$ is not $\textsf{NP}$-complete. 

In this paper, we will investigate the parallel complexity of the \algprobm{Group Isomorphism} problem. There are two key motivations for this. The first comes from \algprobm{GI}, whose best known lower bound is $\textsf{DET}$ \cite{Toran}, which is a subclass of $\textsf{TC}^{1}$. There is a large body of work spanning almost $40$ years on $\textsf{NC}$ algorithms for $\algprobm{GI}$-- see \cite{LevetRombachSieger} for a survey. In contrast, the work on $\textsf{NC}$ algorithms for \algprobm{GpI} is very new, and \algprobm{GpI} is strictly easier than \algprobm{GI} under $\ACz$-reductions \cite{ChattopadhyayToranWagner}. So it is surprising that we know more about parallelizing subclasses of \algprobm{GI} than subclasses of \algprobm{GpI}.

The second key motivation comes from the fact that complexity-theoretic lower bounds for \algprobm{GpI} remain open, even against depth-$2$ $\textsf{AC}$ circuits. In contrast, \algprobm{GpI} is solvable using depth-$4$ $\qACz$ circuits \cite{CGLWISSAC}. Isomorphism testing is in $\textsf{P}$ for a number of families of groups, and closing the gap between $\textsf{P}$ and $\ACz$ for special cases is a key stepping stone in trying to characterize the complexity of the general \algprobm{GpI} problem.

We now turn to stating our main results. %We will first need the following definition.

\subsection{Isomorphism Testing of Coprime Extensions}

Before stating our first main theorem, we will first recall the following notation from Qiao, Sarma, and Tang \cite{QST11}. Let $\mathcal{H}(\mathcal{A}, \mathcal{E})$ be the class of coprime extensions $H \ltimes N$, where $N$ is Abelian and $H$ is elementary Abelian.

\begin{theorem}[cf. {\cite{QST11}}] \label{thm:MainCoprime}
There exists a uniform $\textsf{AC}^{3}$ algorithm that, given groups $G_1, G_2$ by their multiplication tables, decides if $G_1, G_2 \in \mathcal{H}(\mathcal{A}, \mathcal{E})$; and if so, decides if $G_{1} \cong G_{2}$.
\end{theorem}

There has been considerable work on polynomial-time isomorphism tests for coprime and tame extensions. Le Gall \cite{Gal09} exhibited a polynomial-time isomorphism test for groups where the normal Hall subgroup $N$ was Abelian, and the complement $H$ was a cyclic group of order coprime to $|N|$. Subsequently, Qiao, Sarma, and Tang \cite{QST11} extended Le Gall's result, exhibiting polynomial-time isomorphism tests for the cases when (i) $H$ was $O(1)$-generated, and (ii) when $H$ was an arbitrary elementary Abelian group. Babai and Qiao \cite{BQ} subsequently gave a polynomial-time isomorphism test for the more general family of groups with Abelian Sylow towers. Grochow and Qiao \cite{GQ15} later gave a polynomial-time isomorphism test for tame extensions.

The parallel complexity for isomorphism testing of coprime extensions was first investigated by Grochow and Levet \cite{GrochowLevetWL}, who showed that some constant-dimensional Weisfeiler--Leman algorithm identifies, in $O(1)$-rounds, all coprime extensions $H \ltimes N$, where $N$ is Abelian and $H$ is $O(1)$-generated. Consequently, they obtained that the isomorphism problem for this family belongs to $\textsf{L}$. In his thesis, Brachter \cite{BrachterThesis} showed that the class of groups with Abelian Sylow towers has Weisfeiler--Leman dimension $O(\log \log n)$, but did not control for rounds. Consequently, Brachter's work yields $n^{O(\log \log n)}$-time bounds for isomorphism testing of this class. It remains open whether even the class $\mathcal{H}(\mathcal{A}, \mathcal{E})$ has bounded Weisfeiler--Leman dimension. Prior to this paper, it was open whether isomorphism testing of groups in $\mathcal{H}(\mathcal{A}, \mathcal{E})$ was in $\textsf{NC}$. 

\paragraph{Methods.} The polynomial-time isomorphism test for $\mathcal{H}(\mathcal{A}, \mathcal{E})$ \cite{QST11}  crucially leveraged \emph{small} (size $O(\log |G|)$) instances of \algprobm{Linear Code Equivalence}. Thus, in order to obtain our parallel bounds, we will establish the following.

\begin{theorem} \label{thm:MainCodeEq}
Let $m \in O(\log n)$, and let $C_{1}, C_{2}$ be linear codes of length $m$ given by their generator matrices. We can decide equivalence between $C_{1}$ and $C_{2}$, as well as compute the coset of equivalences, using an $\textsf{AC}$ circuit of depth $O((\log^{2} n) \cdot \poly(\log \log n))$ and size $\poly(n)$.
\end{theorem} 

In establishing the bound of $(2+o(1))^m$-time for \algprobm{Linear Code Equivalence}, Babai reduces to isomorphism testing of graphs with $O(m)$ vertices \cite{BCGQ}. When $m \in O(\log n)$, applying Babai's quasipolynomial \algprobm{GI} procedure \cite{BabaiGraphIso} would only yield bounds\footnote{Babai \cite{BabaiGraphIso} established quasipolynomial-time bounds, and the exponent of the exponent was analyzed and improved by Helfgott \cite{HelgottGIAnalysis}.} of $m^{O(\log^2 m)} = (\log n)^{O((\log \log n)^2)}$ runtime, which is not sufficient to obtain bounds of $\textsf{NC}$. 

Instead, we utilize Luks' group-theoretic method \cite{LuksBoundedValence} to handle these small graphs. While Luks' claims of polynomial-time isomorphism testing are for graphs of bounded valence, the graphs we will consider need not satisfy this condition. However, our graphs are \emph{small}, having only $O(\log n)$ many vertices. Thus, we are able to bring to bear a standard suite of $\textsf{NC}$ algorithms for permutation groups \cite{BabaiLuksSeress, McKenzieThesis}. As the group $G$ is specified by its multiplication table and our graphs have $O(\log n) = O(\log |G|)$ many vertices, we are able to implement solutions for each subroutine of \cite{BabaiLuksSeress, McKenzieThesis} using an $\textsf{AC}$ circuit of depth $\poly(\log \log n)$ and size $\poly(n)$. Grochow, Johnson, and Levet \cite{GrochowJohnsonLevet} introduced the notation $\FOPLL$ to refer to the class of languages decidable by uniform $\textsf{AC}$ circuits of depth $\poly(\log \log n)$ and size $\poly(n)$. We combine this suite of \FOPLL permutation group algorithms, as well as the \FOPLL procedure for \textit{small} instances of \algprobm{Coset Intersection} due to Grochow, Johnson, and Levet \cite{GrochowJohnsonLevet}, to bear. 

Our strategy of leveraging a \textit{small} (size $O(\log |G|)$) permutation domain in order to get efficient instances of permutation group algorithms has been a common theme in \algprobm{Group Isomorphism} to obtain polynomial-time isomorphism tests. Crucially, the corresponding \textit{small} instances of \textsf{GI}-hard problems such as \algprobm{Linear Code Equivalence}, \algprobm{Twisted Code Equivalence}, \algprobm{Coset Intersection}, and \algprobm{Set Transporter} can be implemented in time $\poly(|G|)$---see, e.\,g., \cite{BMWGenus2, LewisWilson, IvanyosQ19, QST11, BQ, GQ15, BCGQ, BCQ, GQCoho}. By taking advantage of both a small permutation degree and existing parallel permutation algorithms \cite{BabaiLuksSeress, McKenzieThesis}, we are able to improve the complexity-theoretic upper-bounds for isomorphism testing of the groups in Theorem~\ref{thm:MainCoprime} and Theorem~\ref{thm:MainCoho} (the latter of which will be stated shortly) from $\textsf{P}$ \cite{QST11, GQCoho} to $\textsf{AC}^{3}$.

\subsection{Isomorphism Testing of Central-Radical Groups}
We now turn to our second main result. We first recall that the \emph{solvable radical} is the unique maximal solvable normal subgroup of $G$. We will denote the solvable radical of $G$ as $\text{Rad}(G)$. Here, we will consider \emph{central-radical groups}, which are precisely those groups where the solvable radical $\rad(G) = Z(G)$. In particular, we will restrict attention to the case when $\rad(G) = Z(G)$ is elementary Abelian.

\begin{theorem}[cf. {\cite[Theorem~C]{GQCoho}}] \label{thm:MainCoho}
Let $G_{1}, G_{2}$ be groups given by their multiplication tables. Suppose that $G_1, G_2$ have central, elementary Abelian radicals. Then isomorphism can be decided, and the coset of isomorphisms found, in $\textsf{AC}^{3}$, if either:
\begin{enumerate}[label=(\alph*)]
\item $G_{1}/\rad(G_1)$ is a direct product of non-Abelian simple groups; or
\item $G_{1}/\rad(G_1)$ is a direct product of perfect groups, each of order $O(1)$.
\end{enumerate}
\end{theorem}

\noindent This improves upon the previous polynomial-time bounds due to Grochow and Qiao \cite{GQCoho}.	

Key motivation for isomorphism testing of central-radical groups comes from the following series of characteristic subgroups, known as the \emph{Babai--Beals filtration} \cite{Babai1999GroupsSA} or the \textit{solvable-radical model} (see e.g., \cite{CH03}): 
\[
1 \leq \rad(G) \leq \Soc^{*}(G) \leq \text{PKer}(G) \leq G.
\]
We will now explain the terms on this chain. Let $\pi : G \to G/\rad(G)$ be the natural projection map. $\Soc^{*}(G)$ is the preimage of the socle $\Soc(G/\rad(G))$ (the \emph{socle} of a group $H$ is the direct product of the minimal normal subgroups of $H$), under $\pi$. The group $\Soc^{*}(G)/\rad(G) = \Soc(G/\rad(G))$ is the direct product of non-Abelian simple groups $T_1, \ldots, T_{k}$. The conjugation action of $G$ on $\Soc(G/\rad(G))$ induces a permutation action $\varphi$ on $T_1, \ldots, T_k$. Now $\text{PKer}(G) := \text{Ker}(\varphi)$. Note that Theorem~\ref{thm:MainCoho}(a) considers precisely those groups where $\rad(G) = Z(G)$ is elementary Abelian, and $G = \Soc^{*}(G)$. While these assumptions might feel restrictive, Grochow and Qiao essentially leveraged the full weight of \cite{GQCoho} in order to obtain a polynomial-time isomorphism test, combining deep mathematical insights (group cohomology) with demanding algorithmic techniques (see under Methods, below).

The class of groups without Abelian normal subgroups (a.k.a. \emph{Fitting-free} groups) consists precisely of those groups $G$ where $\rad(G)$ is trivial. There has been considerable work on isomorphism testing of Fitting-free groups. In the multiplication table model, it took a series of two papers \cite{BCGQ, BCQ} to obtain a polynomial-time isomorphism test. Grochow, Johnson, and Levet \cite{GrochowJohnsonLevet} recently improved this bound from $\textsf{P}$ to $\textsf{AC}^{3}$. Fitting-free groups have also received significant attention, from the perspective of the Weisfeiler--Leman algorithm \cite{BrachterSchweitzerWLLibrary, BrachterThesis, GrochowLevetWL, GLDescriptiveComplexity, GrochowJohnsonLevet}.

In the setting of permutation groups and their quotients, Cannon and Holt \cite{CH03} have exhibited an efficient practical algorithm for isomorphism testing of Fitting-free groups. Much of their framework is polynomial-time computable \cite{KantorLuksQuotients, DasThakkarMPD, LevetSrivastavaThakkar}. On the other hand, Grochow, Johnson, and Levet \cite{GrochowJohnsonLevet} exhibited a reduction from \algprobm{Linear Code Equivalence} (and hence, \algprobm{Graph Isomorphism}) to isomorphism testing of Fitting-free groups where $G = \text{PKer}(G)$, given by generating sequences of permutations. The best known bound for \algprobm{Linear Code Equivalence} is $(2+o(1))^{n}$, due to Babai \cite{BCGQ}. Consequently, it seems unlikely that the framework of Cannon and Holt \cite{CH03} achieves a worst-case polynomial-time bound.

We now return to the multiplication table model. In light of the efficient isomorphism tests for Fitting-free groups \cite{BCGQ, BCQ, GrochowJohnsonLevet}, the next step is to consider groups where $\rad(G)$ is non-trivial. In \cite{BCGQ}, the authors posed the problem of handling isomorphism for central-radical groups. Grochow and Qiao \cite{GQCoho} exhibited an $n^{O(\log \log n)}$-time algorithm to decide isomorphism of arbitrary central-radical groups, as well as groups where $\rad(G)$ was elementary Abelian (not necessarily central). Consequently, Grochow and Qiao obtained polynomial-time bounds for a number of important such subclasses. In a different work, Grochow and Qiao exhibited a polynomial-time isomorphism test for groups where $\rad(G)$ had elementary Abelian Sylow towers, and $G/\rad(G)$ was \emph{tame} \cite[Corollary~2]{GQ15}.  Brooksbank, Grochow, Li, Qiao, and Wilson \cite{BGLQW} also exhibited an $n^{O(\log \log n)}$-time algorithm for classes of groups where $\rad(G)$ was a class-$2$ $p$-group of exponent $p \neq 2$, with bounded genus such that $G$ acts on $\rad(G)$ by inner automorphisms. Le Gall and Rosenbaum \cite{GR16} considered the case where $\rad(G)$ was arbitrary, and established that in a precise sense, the projective special linear group constitutes a key obstacle for \algprobm{GpI}. In his dissertation, Brachter \cite{BrachterThesis} showed that when $\rad(G)$ was $O(1)$-generated and $Q$ is an arbitrary Fitting-free group, the groups of the form $Q \ltimes \rad(G)$ Weisfeiler--Leman dimension $O(\log \log n)$. Brachter also established that the class of central-radical groups where $\rad(G)$ is $O(1)$-generated and $G/\rad(G)$ is a direct product of $O(1)$-generated perfect groups, has bounded Weisfeiler--Leman dimension. To the best of our knowledge, there has been no other complexity-theoretic work on isomorphism testing for groups with non-trivial solvable radical.

\paragraph{Methods.} We will now outline the key ideas and techniques involved in proving Theorem~\ref{thm:MainCoho}. As a conceptual starting point, let us first consider the setting of direct products. The Remak--Krull--Schmidt theorem provides that two direct products $G = G_1 \times \cdots \times G_k$ and $H = H_1 \times \cdots \times H_{\ell}$ are isomorpic if and only if $k = \ell$, and for some permutation $\sigma \in \text{Sym}(k)$, $G_{i} \cong H_{\sigma(i)}$ for all $i \in [k]$. Grochow and Qiao established that for the groups in Theorem~\ref{thm:MainCoho}, an analogous result holds \cite[Lemma~7.4]{GQCoho} (recalled as Lemma~\ref{lem:7.4}). 

Let $G_{1}, G_{2}$ be two groups under consideration in Theorem~\ref{thm:MainCoho}. For $i = 1, 2$, write $G_{i}/\rad(G_i) = \prod_{j=1}^{k} T_{i,j}$. Let $\pi_{i} : G_{i} \to G_{i}/\rad(G_{i})$ be the natural projection map. Rather than considering the isomorphism type of the $T_{i,j}$, we will need to consider the isomorphism type of each $\pi_{i}^{-1}(T_{i,j})$. It is then necessary to determine whether there exists a permutation $\sigma \in \text{Sym}(k)$ such that $\pi_{1}^{-1}(T_{1,j})$ and $\pi_{2}^{-1}(T_{2,\sigma(j)})$ are isomorphic. However, there is an additional complication, in that we require a single isomorphism $\alpha : Z(G_{1}) \cong Z(G_{2})$, such that for all $j \in [k]$, $\alpha$ can be extended to an isomorphism of $\pi_{1}^{-1}(T_{1,j})$ and $\pi_{2}^{-1}(T_{2,\sigma(j)})$. Grochow and Qiao handled this step using \emph{small} (size $O(\log |G|)$) instances of \algprobm{Linear Code Equivalence} and \algprobm{Coset Intersection}. In these places, we will use our parallel implementation for \algprobm{Linear Code Equivalence}  (Theorem~\ref{thm:MainCodeEq}), as well as the \FOPLL procedure for \textit{small} instances of \algprobm{Coset Intersection} due to Grochow, Johnson, and Levet \cite{GrochowJohnsonLevet}. 

Still, we will need an additional ingredient. In order to determine the isomorphism types of the $\pi_{i}^{-1}(T_{i,j})$ ($i = 1,2$; $j \in [k]$), Grochow and Qiao utilized \cite[Theorem~6.1]{GQCoho}. Note that \cite[Theorem~6.1]{GQCoho} is quite powerful. Grochow and Qiao obtained one of their main results \cite[Theorem~A/Corollary~6.2]{GQCoho} as a corollary: a bound of $n^{O(\log \log n)}$-time for isomorphism testing of central-radical groups, as well as the same bound for computing the coset of isomorphisms when $\rad(G_i) = Z(G_i)$ was elementary Abelian. We will parallelize \cite[Theorem~6.1]{GQCoho}; for readability, we refer to Theorem~\ref{thm:GQCohoA}, rather than stating the exact theorem here. As a byproduct, we obtain the following complexity-theoretic improvement in isomorphism testing of arbitrary central-radical groups.

\begin{theorem}[cf. {\cite[Theorem~A/Corollary~6.2]{GQCoho}}] 
Isomorphism of central-radical groups of order $n$, given by their multiplication tables, can be decided by $\textsf{AC}$ circuits of depth $O(\log^3 n)$ and size $n^{O(\log \log n)}$. If furthermore the center is elementary Abelian, then we can decide isomorphism and compute the coset of isomorphisms using $\textsf{AC}$ circuits of depth $O(\log^2 n)$ and size $n^{O(\log \log n)}$.	
\end{theorem}

We refer to Section~\ref{sec:Section6} for the proof.

\begin{remark}
Grochow and Qiao discuss key difficulties in extending \cite[Theorem~C]{GQCoho} to a broader class of groups. While  Theorem~\ref{thm:MainCoho} provides a complexity-theoretic improvement over the previous work of Grochow and Qiao, it runs into the same obstacles. We will briefly outline these obstacles here.

As \cite[Theorem~C]{GQCoho} and Theorem~\ref{thm:MainCoho} consider central, elementary Abelian radical groups, it is natural to ask about extending these results to the case when $\rad(G) = Z(G)$ (no assumption of elementary Abelian). Grochow and Qiao showed that \cite[Theorem~C.2]{GQCoho} can be extended to the case when $\rad(G) = Z(G)$, provided that the coset of group code equivalences of length $m$ over cyclic $p$-groups can be found in time $2^{O(m)}$ (see \cite[Observation~7.8]{GQCoho}). It is a longstanding open problem whether a $2^{O(m)}$-time algorithm exists, even to decide whether to such codes are equivalence \cite{BCGQ}.

It is also natural to ask whether we can consider central, elementary Abelian radical groups where $G/\rad(G)$ is a direct product of perfect groups that are each $O(1)$-generated. As non-Abelian finite simple groups are $2$-generated, this class of groups would generalize both of those considered in Theorem~\ref{thm:MainCoho}. However, this would yield codes of length $m \in O(n^{2} \log n)$ in the worst case. A $(2+o(1))^m$ algorithm for \algprobm{Linear Code Equivalence} would then yield a runtime of $(2+o(1))^{O(n^2 \log n)}$, which is worse than the $n^{O(\log n)}$-time bound obtained by generator-enumeration.
\end{remark}

\subsection{Further Related Work}
The work on $\textsf{NC}$ algorithms for \algprobm{GpI} is comparatively nascent compared to that of \algprobm{GI}. Indeed, much of the work involves parallelizing the generator-enumeration strategy (\emph{ibid}), which has yielded bounds of $\textsf{L}$ \cite{TangThesis} for $O(1)$-generated groups. Other families of groups known to admit $\textsf{NC}$ isomorphism tests include Abelian groups \cite{ChattopadhyayToranWagner, GrochowLevetWL, CGLWISSAC}, graphical groups arising from the CFI graphs \cite{WLGroups, CollinsLevetWL, CollinsUndergradThesis}, coprime extensions $H \ltimes N$ where $H$ is $O(1)$-generated and $N$ is Abelian \cite{GrochowLevetWL} (parallelizing a result from \cite{QST11}), groups of almost all orders \cite{CGLWISSAC} (parallelizing \cite{DietrichWilson}), and Fitting-free groups where the number of non-Abelian simple factors of the socle is $O(\log n/\log \log n)$ \cite{GrochowLevetWL} (parallelizing a result from \cite{BCGQ}). Grochow, Johnson, and Levet recently exhibited an $\textsf{AC}^{3}$ isomorphism test for the class of all Fitting-free groups \cite{GrochowJohnsonLevet}. Johnson, Levet, Vojtěchovský, and Widholm \cite{JLVWQuasigroups} recently showed that in the multiplication table model, a fully-refined direct product decomposition of a group can be computed in $\textsf{AC}^{3}$. They also exhibited an $\textsf{NC}$ isomorphism test for central quasigroups. To the best of our knowledge, no other special family of quasigroups that are not groups (beyond $O(1)$-generated quasigroups) are known to admit even a polynomial-time isomorphism test.

\section{Preliminaries}

\subsection{Groups and Codes}
\noindent \textbf{Groups.} All groups will be assumed to be finite. A \emph{Hall} subgroup of a group $G$ is a subgroup $N$ such that $|N|$ and $|G/N|$ are coprime. When a Hall subgroup is normal, we refer to $G$ as a coprime extension. We will consider the following classes of finite groups, using the notation of Qiao, Sarma, and Tang \cite{QST11}. Let $\mathcal{E}$ be the class of  elementary Abelian groups, $\prod \mathcal{E}$ be the class of direct products of elementary Abelian groups, and $\mathcal{A}$ the class of Abelian groups. For classes of finite groups $\mathcal{X}, \mathcal{Y}$, let $\mathcal{H}(\mathcal{X}, \mathcal{Y})$ denote the class of coprime extensions $H \ltimes N$, where $N \in \mathcal{X}$ and $H \in \mathcal{Y}$.

Let $G$ be a group. Given $g, h \in G$, the \textit{commutator} $[g, h] := ghg^{-1}h^{-1}$. For sets $X, Y \subseteq G$, the \textit{commutator subgroup} $[X,Y] := \langle \{ [g, h] : g \in X, h \in Y \} \rangle$. We say that $G$ is \textit{perfect} if $G = [G, G]$, and that $G$ is \textit{centerless} if $Z(G) = 1$. Let $d(G)$ be the minimum size taken over all generating sets of $G$. A \textit{basis} of an Abelian group $A$ is a set of generators $\{a_1, \ldots, a_k\}$ such that $A = \langle a_1 \rangle \times \cdots \times \langle a_k \rangle.$ 

%A \emph{subcoset} of a group $G$ is either the empty set or a coset of a subgroup of $G$.

\noindent \\ \textbf{Codes.} Let $\mathbb{F}$ be a field. Let $\text{Mat}_{n \times m}(\mathbb{F})$ denote the set of $n \times m$ matrices over the field $\mathbb{F}$. $\text{GL}_{m}(\mathbb{F})$ denotes the set of $m \times m$ invertible matrices over $\mathbb{F}$. A \textit{linear code} of length $m$ is a subspace $U \leq \mathbb{F}^{m}$. A $d \times m$ matrix $A$ over $\mathbb{F}$ \textit{generates} the code $U$ if the rows of $A$ span $U$. Let $U,W$ be $d$--dimensional codes of length $m$ over $\mathbb{F}$, generated by $d \times m$ matrices $A, B$, respectively. Then $U$ and $V$ are \emph{equivalent} if and only there exists a permutation matrix $P \in \text{GL}_{m}(\mathbb{F})$ and a matrix $T \in \text{GL}_{d}(\mathbb{F})$ such that $B = TAP$. If $A_1, A_2$ are generator matrices for two codes $U_1, U_2$ respectively, we write $\text{CodeEq}(A_1, A_2)$ for the coset of code equivalences taking $U_1$ to $U_2$.

%%%
\subsection{Group Extensions and Cohomology} \label{sec:GroupExtensions}

In this section, we recall preliminaries concerning group extensions and cohomology from \cite{GQCoho}. We will restrict attention to Abelian cohomology. Given a finite group $G$, we will consider an Abelian normal subgroup $A \trianglelefteq G$ when considering $G$ as an extension of $A$ by $Q := G/A$. Here, we denote this as $A \xhookrightarrow{\iota} G \overset{\pi}{\twoheadrightarrow} Q$, where $\iota : A \to G$ is an injection and $\text{Im}(\iota) = \text{ker}(\pi)$. As $A$ is Abelian, we write the group operation of $A$ additively, even though we denote the group operation of $G$ multiplicatively-- this is standard in this area. Despite that $A \leq G$, we will use these notations in different contexts, and so it should not cause confusion. We refer to $G$ as the \emph{total group}.

Let $\pi : G \to G/A \cong Q$ be the natural projection map. Any function $s : Q \to G$ such that $\pi(s(q)) = q$ for all $q \in Q$ is called a \emph{section} of $\pi$. Any such section gives rise to a function $f_{s} : Q \times Q \to A$ defined by $f_{s}(p,q) = s(p)s(q) \cdot s(pq)^{-1}$. We are free to choose $s(1) = \text{id}_{G}$, and then $f(1,q) = f(q,1) = 0$ for all $q \in Q$. Such sections are called \emph{normalized}. We will assume that all sections are normalized, unless otherwise stated.

For $g \in G$, the conjugation action $\theta_{g} : G \to G$ is defined by $\theta_{g}(x) = gxg^{-1}$. As the operation of $G$ is associative, we have that for all $p, q, r \in Q$:
\[
f_{s}(p,q) + f_{s}(pq,r) = \theta_{p}(f_{s}(q,r)) + f_{s}(p, qr).
\]
This is the \emph{$2$-cocycle identity}.

Any function $f : Q \times Q \to A$ is called a \emph{$2$-cochain}. If $f$ satisfies the $2$-cocycle identity with respect to $\theta$, then $f$ is called a \emph{$2$-cocycle} with respect to $\theta$. Given any homomorphism $\theta : Q \to \Aut(A)$, every $2$-cocycle with respect to $\theta$ arises as $f_{s}$ for some section $s$ of some extension $A \xhookrightarrow{} G \twoheadrightarrow Q$ with respect to $\theta$.

Given any function $u : Q \to A$, the \emph{$2$-coboundary} associated to $u$ is the function $b_{u} : Q \times Q \to A$ defined by $b_{u}(p,q) = u(p) + \theta_{p}(u(q)) - u(pq)$. Any two $2$-cocycles associated to the same extension differ by a coboundary.

The $2$-cochains form an Abelian group $C^{2}(Q,A)$ defined by pointwise addition: $(f+g)(p,q) = f(p,q) + g(p,q)$. Observe that the $2$-cocycle identity is $\mathbb{Z}$-linear. Thus, the $2$-cocycles form a subgroup of $C^{2}(Q,A)$, denoted by $Z^{2}(Q, A, \theta)$. Similarly, the $2$-coboundaries form a subgroup of $Z^{2}(Q, A, \theta)$, denoted by $B^{2}(Q,A,\theta)$.

A \emph{$2$-cohomology class} is a coset of $H^{2}(Q, A, \theta) := Z^{2}(Q,A,\theta)/B^{2}(Q,A,\theta)$. If $f \in Z^{2}(Q,A,\theta)$, we denote the corresponding cohomology class by $[f]$. It follows from the above discussion that each extension $A \xhookrightarrow{} G \twoheadrightarrow Q$ determines a single cohomology class $[f] \in H^{2}(Q,A,\theta)$. We now recall the following definition from Grochow and Qiao:

\begin{definition}[{\cite[Definition~2.1]{GQCoho}}]
Let $A$ be an Abelian group, and $Q$ an arbitrary group. Let $\theta : Q \to \Aut(A)$ be an action, and $f : Q \times Q \to A$ a $2$-cocycle ($f \in Z^{2}(Q,A,\theta)$). The pair $(\theta, f)$ is called \emph{extension data}. Two extension data for the pair $(Q, A)$ are equivalent if they have the exact same action and if the two $2$-cocycles are cohomologous (differ by a coboundary).
\end{definition}

Given an extension $A \xhookrightarrow{} G \twoheadrightarrow Q$, the associated extension data are the action $\theta$ as defined above, and any $2$-cocycle $f_{s}$ for any section $s : Q \to G$. The extension data are in general not unique-- we may choose any representative of the corresponding $2$-cohomology class. Furthermore, if the action is trivial, then this extension is called \emph{central}. If the $2$-cohomology class is trivial, then this extension is called \emph{split}, and $G = P \ltimes A$ for some subgroup $P \leq G$ isomorphic to $Q$.

We now turn to discussing when two different extension data $(\theta_1, f_1), (\theta_2, f_2)$ yield isomorphic total groups. We begin by recalling some additional notation and terminology. Recall that a subgroup $N \leq G$ is \emph{characteristic} if for all $\varphi \in \Aut(G)$, $\varphi(N) = N$. The analogous notion for isomorphisms (rather than automorphisms) is a function $\mathcal{S}$ that assigns to each group $G$ a subgroup $\mathcal{S}(G)$ such that any isomorphism $\varphi : G_1 \to G_2$ restricts to an isomorphism $\varphi|_{\mathcal{S}(G_1)} : \mathcal{S}(G_1) \to \mathcal{S}(G_2)$. We call such a function $\mathcal{S}$ a \emph{characteristic subgroup function}. Most natural characteristic subgroups encountered are characteristic subgroup functions-- for instance, the center $Z(G)$, the commutator $[G,G]$, and the solvable radical $\text{Rad}(G)$.

\begin{definition}[{\cite[Definition~2.2]{GQCoho}}]
Let $A$ be an Abelian group, and let $Q$ be an arbitrary group. Let $(\theta_1, f_1), (\theta_2, f_2)$ be extension data for $A$-by-$Q$. Then the extension data are \emph{pseudo-congruent} if there exists $(\alpha, \beta) \in \Aut(A) \times \Aut(Q)$ such that:
\begin{align}
\theta_{1}(q)(a) = \alpha^{-1}\biggr(\theta_{2}(\beta(q))(\alpha(a)) \biggr) =: \theta_{2}^{\alpha, \beta}(q)(a).
\end{align}
for all $q \in Q, a \in A$, and

\begin{align*}
f_{1}(p,q) = \alpha^{-1}(f_{2}(\beta(p), \beta(q))) + b_{u}(p,q),
\end{align*}
for all $p,q \in Q$ and for some coboundary $b_{u}$. In this case, we write $(\theta_{1}, f_{1}) \cong (\theta_{2}, f_{2})$.
\end{definition}

\begin{lemma}[see e.g., {\cite[Lemma~2.3]{GQCoho}} and {\cite[Section~2.7.4]{Holt2005HandbookOC}}] \label{lem:GQMainLemma}
Let $\mathcal{S}$ be a characteristic subgroup function. Let $G_{1}, G_{2}$ be groups, such that $\mathcal{S}(G_1)$ and $\mathcal{S}(G_2)$ are both Abelian. Then $G_1 \cong G_2$ if and only if both of the following conditions hold:
\begin{enumerate}
\item $\mathcal{S}(G_1) \cong \mathcal{S}(G_2)$ and $G_{1}/\mathcal{S}(G_1) \cong G_{2}/\mathcal{S}(G_2)$.
\item $(\theta_1, f_1) \cong (\theta_2, f_2)$, where $(\theta_i, f_i)$ is the extension data of $A \xhookrightarrow{} G_i \twoheadrightarrow Q$ ($i = 1, 2$).
\end{enumerate}
\end{lemma}

Note that if the $2$-cohomolohy class is trivial (in which case, the extension splits), Lemma~\ref{lem:GQMainLemma} yields Taunt's Lemma as a corollary: 

\begin{lemma}[{Taunt \cite{Taunt1955}}]  \label{lem:Semidirect}
Let $G = H \ltimes_{\theta} N$ and $\widehat{G} = \hat{H} \ltimes_{\hat{\theta}} \widehat{N}$. If $\alpha : H \to \widehat{H}$ and $\beta : N \to \widehat{N}$ are isomorphisms such that for all $h \in H$ and all $n \in N$,
\[
\widehat{\theta}_{\alpha(h)}(n) = (\beta \circ \theta_{h} \circ \beta^{-1})(n),
\]
then the map $(h, n) \mapsto (\alpha(h), \beta(n))$ is an isomorphism of $G \cong \widehat{G}$. Conversely, if $G$ and $\widehat{G}$ are isomorphic and $|H|$ and $|N|$ are coprime, then there exists an isomorphism of this form.
\end{lemma}

%%%
\subsection{Computational Complexity} \label{sec:Complexity}
We assume that the reader is familiar with standard complexity classes such as $\textsf{P}, \textsf{NP}, \textsf{L}$, and $\textsf{NL}$. Denote by $\mathsf{FL}$ the class of logspace computable functions. For a standard reference on circuit complexity, see \cite{VollmerText}. We consider Boolean circuits using the gates \textsf{AND}, \textsf{OR}, \textsf{NOT}, and \textsf{Majority}, where $\textsf{Majority}(x_{1}, \ldots, x_{n}) = 1$ if and only if $\geq n/2$ of the inputs are $1$. Otherwise, $\textsf{Majority}(x_{1}, \ldots, x_{n}) = 0$. 

In this paper, we will consider $\textsf{DTIME}(\log^c n)$-uniform circuit families $(C_{n})_{n \in \mathbb{N}}$, for some fixed $c \geq 1$. For this,
one encodes the gates of each circuit $C_n$ by bit strings of length $O(\log^c n)$. Then the circuit family $(C_n)_{n \geq 0}$
is called \emph{\textsf{DTIME}$(\log^c n)$-uniform}  if (i) there exists a deterministic Turing machine that computes for a given gate $u \in \{0,1\}^*$
of $C_n$ ($|u| \in O(\log^c n)$) in time $O(\log^c n)$ the type of gate $u$, where the types are $x_1, \ldots, x_n$, \textsf{NOT}, \textsf{AND}, \textsf{OR}, or \textsf{Majority} gates,
and (ii) there exists a deterministic Turing machine that decides for two given gates $u,v \in \{0,1\}^*$ of $C_n$ ($|u|, |v| \in O(\log^c n)$) and a binary encoded integer $i$ with $O(\log^c n)$ many bits in time $O(\log^c n)$ whether $u$ is the $i$-th input gate for $v$. When $c = 1$, this notion of uniformity is referred to as $\textsf{DLOGTIME}$-uniformity. For circuit families of size $\poly(n)$, we will use $\textsf{DLOGTIME}$-uniformity.

\begin{definition}
Fix $k \geq 0$. We say that a language $L$ belongs to (uniform) $\textsf{NC}^{k}$ if there exist a (uniform) family of circuits $(C_{n})_{n \in \mathbb{N}}$ over the $\textsf{AND}, \textsf{OR}, \textsf{NOT}$ gates such that the following hold:
\begin{itemize}
\item The $\textsf{AND}$ and $\textsf{OR}$ gates take exactly $2$ inputs. That is, they have fan-in $2$.
\item $C_{n}$ has depth $O(\log^{k} n)$ and uses (has size) $n^{O(1)}$ gates. Here, the implicit constants in the circuit depth and size depend only on $L$.

\item $x \in L$ if and only if $C_{|x|}(x) = 1$. 
\end{itemize}
\end{definition}

\noindent The complexity class $\AC^{k}$ is defined analogously as $\textsf{NC}^{k}$, except that the $\textsf{AND}, \textsf{OR}$ gates are permitted to have unbounded fan-in.
That is, a single $\textsf{AND}$ gate can compute an arbitrary conjunction, and a single $\textsf{OR}$ gate can compute an arbitrary disjunction. The class $\SAC^k$ is defined analogously, in which the $\textsf{OR}$ gates have unbounded fan-in but the $\textsf{AND}$ gates must have fan-in $2$.
The complexity class $\textsf{TC}^{k}$ is defined analogously as $\AC^{k}$, except that our circuits are now permitted $\textsf{Majority}$ gates of unbounded fan-in.
We also allow circuits to compute functions by using multiple output gates. 

For every $k$, the following containments are well-known:
\[
\textsf{NC}^{k} \subseteq \SAC^k \subseteq  \AC^{k} \subseteq \textsf{TC}^{k} \subseteq \textsf{NC}^{k+1}.
\]

\noindent In the case of $k = 0$, we have that:
\[
\textsf{NC}^{0} \subsetneq \AC^{0} \subsetneq \textsf{TC}^{0} \subseteq \textsf{NC}^{1} \subseteq \LogSpace \subseteq \textsf{NL} \subseteq \textsf{SAC}^{1} \subseteq \AC^{1}.
\]

\noindent We note that functions that are $\textsf{NC}^{0}$-computable can only depend on a bounded number of input bits. Thus, $\textsf{NC}^{0}$ is unable to compute the $\textsf{AND}$ function. It is a classical result that $\AC^{0}$ is unable to compute \algprobm{Parity} \cite{FSS}. The containment $\textsf{TC}^{0} \subseteq \textsf{NC}^{1}$ (and hence, $\textsf{TC}^{k} \subseteq \textsf{NC}^{k+1}$) follows from the fact that $\textsf{NC}^{1}$ can simulate the unbounded fan-in \textsf{Majority} gate.

The complexity class $\textsf{FOLL}$ is the set of languages decidable by uniform $\textsf{AC}$ circuits of depth $O(\log \log n)$ and polynomial-size \cite{BKLM}. We will use a slight generalization of this: we use $\FOPLL$ to denote the class of languages $L$ decidable by uniform $\textsf{AC}$ circuits of depth $O((\log \log n)^c)$ for some $c$ that depends only on $L$ \cite{GrochowJohnsonLevet}. It is known that $\textsf{AC}^{0} \subsetneq \FOPLL \subsetneq \textsf{AC}^{1}$, the former by a simple diagonalization argument on top of Sipser's result \cite{SipserBorel}, and the latter because the \textsf{Parity} function is in $\mathsf{AC}^1$ but not $\FOPLL$ (nor any depth $o(\log n / \log \log n)$). \FOPLL\ cannot contain any complexity class that can compute \algprobm{Parity}, such as $\mathsf{TC}^0, \mathsf{NC}^1, \mathsf{L}, \mathsf{NL}$, or $\mathsf{SAC}^1$, and it remains open whether any of these classes contain \FOPLL.

We will also be interested in $\mathsf{NC}$ circuits of quasipolynomial size (i.\,e., $2^{O(\log^k n)}$ for some constant $k$). For a circuit class $\mathcal{C} \subseteq \mathsf{NC}$, the analogous class permitting a quasipolynomial number of gates is denoted $\mathsf{quasi}\mathcal{C}$. Note that \DLOGTIME uniformity does not make sense for $\textsf{quasiNC}$, as we cannot encode gate indices using $O(\log n)$ bits. Instead, we will use $\textsf{DPOLYLOGTIME}$-uniformity for $\mathsf{quasiNC}$ \cite{BarringtonQuasipolynomial,FerrarottiGonzalezScheweTurull}.

\subsection{Permutation Group Algorithms} \label{sec:PermutationGroups}

We will consider the permutation group model, in which groups are specified succinctly by a sequence of permutations from $\text{Sym}(m)$. The computational complexity for the permutation group model will be measured in terms of $m$. We will first recall some key concepts concerning permutation groups. We refer to the textbook by Dixon \& Mortimer \cite{DixonMortimer} for many of the concepts around permutation groups that we briefly review here. A permutation group $G \leq \text{Sym}(m)$ is \emph{transitive} if the permutation domain $[m]$ is a single $G$-orbit. An equivalence relation $\sim$ on $[m]$ is $G$-invariant if $x \sim y \Leftrightarrow x^g \sim y^g$ for all $x,y \in [m], g \in G$. The discrete equivalence---in which all elements are pairwise inequivalent---and indiscrete equivalence---in which all elements are equivalent---are considered trivial. The equivalence classes of any $G$-invariant equivalence relation are called \emph{blocks} of $G$. The singleton sets and the whole set $[m]$ are considered trivial blocks; any other blocks, if they exist, are non-trivial. $G$ is \emph{primitive} if it has no non-trivial blocks, or equivalently, has no non-trivial $G$-invariant equivalence relations. 

Equivalently, a block for $G$ is a subset $\Delta \subseteq [m]$ such that for every $g \in G$, $\Delta^g$ is either disjoint from $\Delta$ or equal to $\Delta$. A \emph{system of imprimitivity} for $G$ is a collection of blocks that partition $[m]$ (equivalently, the equivalence classes of a $G$-invariant equivalence relation). If $G$ is transitive, every system of imprimitivity arises as $\{\Delta^g : g \in G\}$ for some block $\Delta$. A block is \emph{minimal} if it is non-trivial and contains no proper subset that is also a non-trivial block. Two distinct minimal blocks can intersect in at most one point.

We will now recall a standard suite of problems with known $\textsf{NC}$ solutions in the setting of permutation groups. In our setting, $m$ will be \textit{small} ($m \in O(\text{polylog}(n))$) relative to the overall input size $n$, which will yield $\FOPLL$ bounds for these problems.

\begin{lemma} \label{PermutationGroupsNC}
Let $c \in \mathbb{Z}^{+}$, and let $m \in O(\log^c n)$. Let $G \leq \text{Sym}(m)$ be given by a sequence $S$ of generators. The following problems are in $\FOPLL$ relative to $n$, that is, they have uniform $\mathsf{AC}$ circuits of depth $\poly(\log \log n)$ and $\poly(n)$ size:
\begin{enumerate}[label=(\alph*)]
\item Compute the order of $G$.
\item Decide whether a given permutation $\sigma$ is in $G$; and if so, exhibit a word $\omega$ such that $\sigma = \omega(S)$. 

\item Find the kernel of any action of $G$.
\item Find the pointwise stabilizer of $B \subseteq [m]$. 
%\item Find the normal closure of any subset of $G$. \msay{Delete?}
%\item Compute $\FOPLL$-efficient Schreier generators. In particular, we may take the corresponding chain of subgroups $G = G_0 \geq G_1 \geq \cdots \geq G_r = 1$ such that $r \in O(\log n)$ and $G_{i}$ is the point-stabilizer of the first $i$ points.  \msay{Delete?}

\item Given a list of generators $S$ for $G$, compute a minimal (non-redundant) of generators $S'$ for $G$ of size $\poly(m)$. 

%\item Let $p$ be a prime divisor of  $|G|$. Compute generators for a Sylow $p$--subgroup of $G$. \msay{Delete?}

\item Computing a non-trivial, minimal block system for $G$, if one exists; otherwise, a report that $G$ is primitive.
\item Computing the orbits of $G$.

\item \algprobm{Coset Intersection}: Given $G, H \leq \text{Sym}(m)$ and $x, y \in \text{Sym}(m)$, compute $Gx \cap Hy$. \label{CosetInt}
\end{enumerate}
\end{lemma}

\begin{proof}
Each of these problems in (a)--(g) is known to be solvable using a circuit of depth $O(\text{polylog}(m))$ and size $O(\poly(m))$: see \cite{BabaiLuksSeress} for (a)--(e), \cite{McKenzieThesis} for (f)--(g). The $\FOPLL$ bound follows from the fact that $m \in O(\log^c n)$. 

For (h), the \FOPLL bound was established in \cite{GrochowJohnsonLevet}\footnote{Grochow, Johnson, and Levet only stated this bound for the case when $m \in O(\log n)$. However, a careful analysis of their proof shows that the \FOPLL bound holds when $m \in O(\log^c n)$.}. 
\end{proof}

Finally, we consider the \algprobm{Transversal} problem, which takes as input $H < G \leq \text{Sym}(m)$ such that $[G : H] \in m^{O(1)}$ and asks for a left transversal\footnote{Kantor, Luks, and Mark ask for a right transversal. The proof of \cite[Proposition~3.13]{KantorLuksMark} is easily modified to yield a left transversal without affecting the complexity.} for $H$ in $G$. Kantor, Luks, and Mark \cite[Proposition~3.13]{KantorLuksMark} exhibited an $\textsf{NC}$ algorithm for \algprobm{Transversal} when $|G| \in m^{O(\log m)}$. We will instead consider $\algprobm{Transversal}$ when $m \in O(\log n)$ and $[G : H] \in (\log n)^{O(\log \log n)}$. A careful analysis of \cite[Proposition~3.13]{KantorLuksMark} yields the following for our setting:

\begin{lemma}[cf. {\cite[Proposition~3.13]{KantorLuksMark}}] \label{lem:Transversals}
Let $m \in O(\log n)$, and let $H < G \leq \text{Sym}(m)$. Suppose that $[G : H] \in (\log n)^{O(\log \log n)}$. We can solve $\algprobm{Transversal}$ using an $\textsf{AC}$ circuit of depth $O((\log n) \cdot \poly(\log \log n))$.
\end{lemma}

\begin{proof}
As $G \leq \text{Sym}(m)$, $|G| \leq m!$. From the proof of \cite[Proposition~3.13]{KantorLuksMark}, the recursion depth is $\log |G| \sim m \log m \in O((\log n) \cdot (\log \log n))$. At each recursive call, we use subroutines for deciding subgroup membership and computing the pointwise stabilizer. By Lemma~\ref{PermutationGroupsNC}(b) and (d), respectively, each of these subroutines is \FOPLL-computable. The result now follows.
\end{proof}

%%%%%
\subsection{Representation Theory of Finite Groups}

We recall key preliminaries concerning the representation theory of finite groups. For a general reference, see \cite{serre1977}.

Let $G$ be a finite group, and let $V$ be a vector space. A \textit{representation} of $G$ over $V$ is a group homomorphism $\varphi : G \to \text{GL}(V)$. The \textit{trivial representation} maps every group element to the identity matrix in $\text{GL}(V)$. If $V \cong \mathbb{F}^{d}$, for some field $\mathbb{F}$ and some $d \in \mathbb{N}$, a homomorphism $\varphi : G \to \text{GL}_{d}(\mathbb{F})$ is called a \textit{representation} of $G$ over $\mathbb{F}$ of dimension $d$. Let $\varphi : G \to \text{GL}(V)$ be a representation. A subspace $L \leq V$ is \textit{invariant} with respect to $\varphi$ if for all $g \in G$, $\varphi_{g}(L) = L$. Note that $0, V$ are the \textit{trivial invariant subspaces}. A representation without any non-trivial invariant subspaces is called an \textit{irreducible representation}. We may test whether a representation is irreducible using the following.

\begin{proposition} \label{prop:Irreducibility}
Let $\varphi : G \to \text{GL}(V)$ be a representation. We have that $\varphi$ is irreducible if and only if for all non-zero $v \in V$, $\langle v^{g} | g \in G \rangle = V$.
\end{proposition}

Fix a field $\mathbb{F}$. For the rest of this section, we consider representations over $\mathbb{F}$. Let $\varphi : G \to \text{GL}(V)$ and $\rho : G \to \text{GL}(W)$ be representations. The \textit{direct sum} $\varphi \oplus \rho$ is a representation of $G$ over $V \oplus W$, defined as $(\varphi \oplus \rho)_{g}(u + v) = \varphi_{g}(u) + \rho_{g}(v)$, for all $g \in G$. A representation is \textit{completely reducible} if it is the direct sum of irreducible representations. Maschke's theorem states that if the characteristic of $\mathbb{F}$ is $0$ or coprime to $|G|$, then any representation of $G$ over $\mathbb{F}$ is completely reducible.	

Two representations $\varphi, \psi : G \to \text{GL}(V)$ are \textit{equivalent} if there exists an invertible linear transformation $T : V \to V$ such that $\varphi(g) = T \psi(g) T^{-1}$ for all $g \in G$. Suppose that $\varphi, \psi$ are completely reducible. Write $\varphi = \iota_{1}^{d_{1}} \oplus \cdots \oplus \iota_{\ell}^{d_{\ell}}$ and $\psi = \iota_{1}^{e_{1}} \oplus \cdots \oplus \iota_{\ell}^{e_{\ell}}$, where the $\iota_{i}$'s are irreducible and pairwise inequivalent, and each $d_{i}, e_{i} \geq 0$. We have that $\varphi$ and $\psi$ are equivalent if and only if $d_{i} = e_{i}$ for each $i \in [\ell]$.

The \textit{character} of $\varphi$, denoted $\chi_{\varphi} : G \to \mathbb{F}$, is the function that maps $\chi_{\varphi}(g) = \text{tr}(\varphi(g))$. If the characteristic of $\mathbb{F}$ does not divide $|G|$, then two representations of $G$ are equivalent if and only if their characters are the same. This yields the following.

\begin{lemma} \label{lem:CompareCharacters}
Let $G$ be a group of order $n$ specified by its multiplication table, and let $p \leq n$ be a prime that does not divide $n$. Let $k \in O(\log n)$, and $\varphi, \psi : G \to \text{GL}_{k}(\mathbb{Z}_{p})$ be representations given explicitly by associating with each $g \in G$, the corresponding matrix $\varphi(g)$ (respectively, $\psi(g)$). We can decide if $\varphi, \psi$ are equivalent in $\textsf{AC}^{0}$.
\end{lemma}

\begin{proof}
As $p$ does not divide $n$, $\varphi$ and $\psi$ are equivalent if and only if their characters are the same. For a given $g \in G$, we may evaluate $\text{tr}(\varphi(g))$ as follows. First, we note that as $p \leq n$, each entry in $\varphi(g)$ is represented using $O(\log n)$ bits. We first compute $\sum_{i=1}^{k} \varphi(g)_{ii}$ and $\sum_{i=1}^{k} \psi(g)_{ii}$. As $k \in O(\log n)$, we are adding $O(\log n)$-many $O(\log n)$-bit integers. This step is $\textsf{AC}^{0}$-computable (cf. \cite{VollmerText}). Note that the resulting sum uses only $O(\log n)$-many bits. Thus, we may reduce the resulting sum modulo $p$ in $\textsf{AC}^{0}$. We can then test whether $\sum_{i=1}^{k} \varphi(g)_{ii} = \sum_{i=1}^{k} \psi(g)_{ii}$ in $\ACz$. 
\end{proof}

We now show that we can efficiently decompose a representation.

\begin{proposition} \label{prop:DecomposeRepresentation}
Let $p, q$ be distinct primes. Let $G = \mathbb{Z}_{q}^{\ell} \ltimes_{\varphi} \mathbb{Z}_{p}^{k}$ be  given by its multiplication table. In $\textsf{AC}^{1}$, we can list the irreducible components of $\varphi$ and group them by equivalence type. In particular, each irreducible component is returned as a function of the form $\tau : \mathbb{Z}_{q}^{\ell} \to \text{GL}_{k_{\tau}}(\mathbb{Z}_{p})$, where $k_{\tau}$ is the dimension of $\tau$. 
\end{proposition}

\begin{proof}
Let $v \in \mathbb{Z}_{p}^{k}$ be non-zero. In $\textsf{FL}$, we may write down each $\varphi_{g}(v) = v^{g}$. Then in $\textsf{FL} \cap \textsf{FOLL}$, we may write down $W = \langle v^{g} : g \in G \rangle$ using a membership test \cite{BKLM}. Now for each $w \in W$, we may write down $W_{w}' = \langle w^{g} : g \in G \rangle$ and test whether $W = W'$ in $\textsf{L} \cap \textsf{FOLL}$. By Proposition~\ref{prop:Irreducibility}, $W$ is irreducible if and only if for all $w \in W$, $W_{w}' = W$.

Now for each invariant subspace $W$, we may in $\textsf{AC}^{1}$, fix a basis for that subspace \cite[Proposition~4.5]{JLVWQuasigroups}. Now for each $g \in \mathbb{Z}_{q}^{\ell}$, we consider the conjugation action of $g$ on the basis of $W$. This allows us to compute the matrix representation $M_{g}$ of $g$ on $W$. Let $\tau_{W} : \mathbb{Z}_{q}^{\ell} \to \text{GL}_{\text{dim}(W)}(\mathbb{Z}_{p})$ be the representation sending $g \mapsto M_{g}$. By Lemma~\ref{lem:CompareCharacters} we can, in $\textsf{AC}^{0}$, decide when two  representations are equivalent. Thus, in $\textsf{AC}^{0}$, we can group the equivalent irreducible representations. 

Our algorithm is $\textsf{AC}^{1}$-computable, as desired. The result now follows.
\end{proof}

%%%%%%%
\section{Isomorphism Testing of Small Graphs in Parallel}

In this section, we will establish the following.

\begin{theorem}[cf. {\cite{LuksBoundedValence}}] \label{thm:LuksBoundedGenerator}
Fix $n \in \mathbb{N}$. Let $X,Y$ be graphs with $m \in O(\log n)$ vertices. We can decide whether $X$ and $Y$ are isomorphic, using an $\textsf{AC}$ circuit of depth $O((\log n)^2 \cdot \poly(\log \log n))$ and size $\poly(n)$.
\end{theorem}

In order to establish Theorem~\ref{thm:LuksBoundedGenerator}, we essentially parallelize the previous work of Luks for isomorphism testing for graphs of bounded valence \cite{LuksBoundedValence}. While our graphs here need not have bounded valence, they are \emph{small}. We will crucially take advantage of this, in tandem with a suite of efficient parallel algorithms for permutation groups (see Section~\ref{sec:PermutationGroups}).

Fix an edge $e$ of $X$. We compute $\Aut_{e}(X)$, the automorphism group of $X$ that fixes $e$. Following \cite{LuksBoundedValence}, define $X_{r}$ to be the subgraph of $X$ consisting of all vertices and all edges of $X$, which appear in paths of length $\leq r$ through $e$. So $X_{1} = e$ and $X_{m-1} = X$. We will compute $\Aut_{e}(X_{r})$ for each $r = 1, \ldots, m-1$. The groups are related via the homomorphisms:
\[
\pi_{r}: \Aut_{e}(X_{r+1}) \to \Aut_{e}(X_{r}),
\]

\noindent where $\pi_{r}(\sigma)$ is the restriction of $\sigma$ to $X_{r}$. Thus, given (generators for) $\Aut_{e}(X_{r})$, determining generators for $\Aut_{e}(X_{r+1})$ reduces to two problems:
\begin{itemize}
\item Find a set $\mathcal{R}$ of generators for $\text{Ker}(\pi_{r})$.
\item Find a set $\mathcal{S}$ of generators for $\text{Im}(\pi_{r})$.
\end{itemize}

Let $\Delta(X)$ be the maximum degree of $X$. Let $A \subseteq 2^{V(X_{r})} \setminus \{ \emptyset\}$ contain those subsets of size at most $\Delta$. Define $f : V(X_{r+1}) \setminus V(X_{r}) \to A$ by:
\[
f(v) := \{ w \in V(X_{r}) : vw \in E(X) \}.
\]

\begin{proposition} \label{prop:ComputeKernel}
We can compute $\mathcal{R}$ in $\ACz$.
\end{proposition}

\begin{proof}
Luks \cite[Section~3.1]{LuksBoundedValence} established that $\text{Ker}(\pi_{r}) = \prod_{a \in A} \text{Sym}(f^{-1}(a))$. As $|V(X)| \in O(\log n)$, we have that $\Delta \in O(\log n)$ and $|A| \in \poly(n)$. So we can compute $\Delta$, and hence $A$ in $\textsf{AC}^{0}$. Furthermore, we can compute $f^{-1}(a)$ in $\ACz$ for each $a \in A$. Now the symmetric group is $2$-generated. So in $\ACz$, we may write down standard generators for each $\text{Sym}(f^{-1}(a))$. The result now follows.
\end{proof}

It remains to compute $\mathcal{S}$. Our key result will thus be as follows:

\begin{proposition} \label{prop:ComputeImage}
Take the same assumptions as Theorem~\ref{thm:LuksBoundedGenerator} and fix $r \in [m-1]$. Suppose we are given generators for $\Aut_{e}(X_{r})$. We can compute $\mathcal{S}$ using an $\textsf{AC}$ circuit of depth $O((\log n) \cdot \poly(\log \log n))$ and size $\poly(n)$.
\end{proposition}

Proposition~\ref{prop:ComputeImage} immediately yields Theorem~\ref{thm:LuksBoundedGenerator}.

\begin{proof}[Proof of Theorem~\ref{thm:LuksBoundedGenerator}]
We proceed for $m-1 \in O(\log n)$ times iterations. At each iteration, we apply Proposition~\ref{prop:ComputeKernel} and Proposition~\ref{prop:ComputeImage}. The result now follows.
\end{proof}

The remainder of this section will be devoted to the proof of Proposition~\ref{prop:ComputeImage}.

\noindent Luks \cite{LuksBoundedValence} previously established that $\sigma \in \Aut_{e}(X_{r})$ is in $\text{Im}(\pi_{r})$ if and only if, for each $0 \leq s \leq m-1$, $\sigma$ stabilizes the set of fathers of $s$-tuples:
\[
A_{s} := \{ a \in A : |f^{-1}(a)| = s \},
\]

\noindent as well as the set $A'$ of new edges. Color $A$ accordingly with $2m$ colors (see e.g., \cite[Page~49]{LuksBoundedValence}). We need to now find the automorphisms in $G = \Aut_{e}(X_{r})$ acting on $A$, that preserve the edge colors. We now recall some notions from Luks \cite{LuksBoundedValence}.

\begin{definition} \label{def:ColorPreservingAutomorphisms}
Let $A$ be a colored set, with coloring $\chi$. Let $B \subseteq A$ and $K \leq \text{Sym}(A)$. The set of permutations in $K$ that preserve the colors in $B$ is:
\[
\mathcal{C}_{B}(K) := \{ \sigma \in K : \text{for all } b\in B, \chi(\sigma(b)) = \chi(b) \}.
\]
\end{definition}

\noindent Observe that:
\begin{align*}
&\mathcal{C}_{B}(K \cup K') = \mathcal{C}_{B}(K) \cup \mathcal{C}_{B}(K'), \text{ and} \\
%&\mathcal{C}_{B \cup B'}(K) = \mathcal{C}_{B}(\mathcal{C}_{B'}(K)).
&\mathcal{C}_{B \cup B'}(K) = \mathcal{C}_{B}(K) \cap \mathcal{C}_{B'}(K).
\end{align*}
Note that if $\mathcal{C}_{B}(\sigma B)$ is non-empty, then it is a left-coset of the subgroup $\mathcal{C}_{B}(G) \leq K$ \cite[Lemma~2.4]{LuksBoundedValence}.

\begin{remark}
Luks \cite{LuksBoundedValence} uses that $\mathcal{C}_{B \cup B'}(K) = \mathcal{C}_{B}(\mathcal{C}_{B'}(K))$. As $m \in O(\log n)$, we can compute $\mathcal{C}_{B}(K) \cap \mathcal{C}_{B'}(K)$ in $\FOPLL$ (Lemma~\ref{PermutationGroupsNC}\ref{CosetInt}), and we will later take advantage of this.
\end{remark}

We now turn to solving the \algprobm{Color Automorphism} problem. For $k \in \mathbb{N}$, let $\Gamma_{k}$ be the class of groups $G$ where all the non-Abelian composition factors of $G$ are subgroups of $\text{Sym}(k)$.

\begin{definition}
The \algprobm{Color Automorphism} problem is defined as follows:
\begin{itemize}
\item \textsf{Instance:} Generators for a subgroup $G \leq \text{Sym}(A)$ with $G \in \Gamma_{k}$, a $G$-stable subset $B$, and $\sigma \in \text{Sym}(A)$.

\item \textsf{Solution:} The set $\mathcal{C}_{B}(\sigma G)$.
\end{itemize}
\end{definition}

\noindent We will consider the \algprobm{Color Automorphism} problem in the case when $k = m \in O(\log n)$. In order to solve the \algprobm{Color Automorphism} problem, we proceed via a divide-and-conquer procedure. We have the following cases in the recursion:
\begin{itemize}
\item \textbf{Case 1 (Base Case):} The recursion bottoms out when $|B| = 1$. In this case, we compute the \algprobm{Pointwise Stabilizer} of $\mathcal{C}_{B^{\sigma^{-1}}}(G)$ in \FOPLL using Lemma~\ref{PermutationGroupsNC}(d). Now $\sigma\mathcal{C}_{B^{\sigma^{-1}}}(G) = \mathcal{C}_{B}(\sigma G)$.

\item \textbf{Case 2 (Intransitive Case):} Suppose that $B$ is the disjoint union of $G$-stable subsets $Z_1, \ldots, Z_k$. We use Lemma~\ref{PermutationGroupsNC}(g) to, in \FOPLL, break $B$ up into orbits $Z_1, \ldots, Z_k$. In particular, for each $i \in [k]$, $G$ acts transitively on $Z_i$. We now have that:
\[
\mathcal{C}_{B}(\sigma G) = \bigcap_{i=1}^{k} \mathcal{C}_{Z_{i}}(\sigma G).
\] 

We may thus compute each $\mathcal{C}_{Z_{i}}(\sigma G)$ in parallel. Note that as $G$ acts transitively on each $Z_{i}$, the recursive calls for $\mathcal{C}_{Z_{i}}(\sigma G)$ fall under Case 1 or Case 3. In order to compute the intersection, we use a binary tree circuit of depth $O(\log k) \leq O(\log |X|) \leq O(\log \log n)$. At each node of the binary tree, we utilize the $\FOPLL$ algorithm for $\algprobm{Coset Intersection}$ (Lemma~\ref{PermutationGroupsNC}\ref{CosetInt}). Thus, the total non-recursive work in this case is $\FOPLL$.

\item \textbf{Case 3 (Transitive Case):} If $|B| > 1$ and $B$ is not the disjoint union of at least two $G$-stable subsets, then we find a minimal $G$-block system in $B$:
\[
\Phi = \{ B_{1}, B_{2}, \cdots, B_{\ell} \}.
\]

As $|X| \in O(\log n)$, we have by Lemma~\ref{PermutationGroupsNC}(f) that such a block system is \FOPLL-computable. Furthermore, by Lemma~\ref{PermutationGroupsNC}(c), we may compute the kernel $K$ of the $G$-action on $\Phi$ in $\FOPLL$. From the work of Babai, Cameron, and Palfy \cite{BabaiCameronPalfy}, we have that $[G : K] \in (\log n)^{O(\log \log n)}$. Thus, we may write $G$ as a union of $(\log n)^{O(\log \log n)}$ cosets:
\[
G = \bigcup_{i=1}^{[G : K]} \tau_{i}K.
\]
This step is computable using an $\textsf{AC}$ circuit of depth $O((\log n)\poly(\log \log n))$ and size $\poly(n)$, using Lemma~\ref{lem:Transversals}. Now the problem breaks up as follows:
\[
\mathcal{C}_{B}(\sigma G) = \bigcup_{i=1}^{[G:K]} \mathcal{C}_{B}(\sigma \tau_{i}K).
\]

We now, in parallel, recursively compute $\mathcal{C}_{B}(\sigma \tau_{i}K)$ for each $1 \leq i \leq [G:K]$. As $K$ fixes (setwise) each block of $\Phi$, the computation for each $\mathcal{C}_{B}(\sigma \tau_{i}K)$ falls under Case 1 or Case 2. We may, in parallel, recursively compute each $\mathcal{C}_{B}(\sigma \tau_{i}K)$. Note that each recursive call deals with subsets of size $|B|/\ell$. 

Note that when we recombine the cosets, we must take care to ensure that we only have $\text{poly}(n)$ generators. To accomplish this, we use Lemma~\ref{PermutationGroupsNC}(e). As our permutation domain has size $O(\log n)$, each application of Lemma~\ref{PermutationGroupsNC}(e) is \FOPLL-computable. We use a binary tree circuit of depth $\log [G:K] \in O((\log \log n)^2)$ to compute $\mathcal{C}_{B}(\sigma G)$ from the $\mathcal{C}_{B}(\sigma \tau_{i}K)$ ($1 \leq i \leq [G : K]$). This yields a circuit of depth $\poly(\log \log n)$ and size $\poly(n)$ to compute $\mathcal{C}_{B}(\sigma G)$.
\end{itemize}

\noindent \\ \textbf{Complexity Analysis.} First, observe that a recursive call at the intransitive case (Case 2) results in either the base case (Case 1) or the transitive case (Case 3). Similarly, a recursive call at the transitive case (Case 3) results in either the base case (Case 1) or the intransitive case (Case 3). The recursive calls in the transitive case reduce the size of the problem by at least $1/2$. Thus, the recursion tree thus has height $O(\log m) = O(\log \log n)$. Each level of the recursion tree is computable using an $\textsf{AC}$ circuit of depth $O((\log n)\poly(\log \log n))$ and size $\poly(n)$. Thus, in total, we require an $\textsf{AC}$ circuit of depth $O((\log n)\poly(\log \log n))$ and size $\poly(n)$, as desired. This completes the proof of Proposition~\ref{prop:ComputeImage}.

\section{Linear Code Equivalence for Small Codes in Parallel}

In this section, we establish the following:

\begin{theorem}[cf. {\cite[Theorem~7.1]{BCGQ}}] \label{thm:LinearCodeEquivalence}
Let $m \in O(\log n)$, and let $C_1, C_2$ be linear codes of length $m$ given by their generator matrices. We can decide equivalence between $C_1$ and $C_2$, as well as compute the coset of equivalences, using an $\textsf{AC}$ circuit of depth $O((\log^2 n) \cdot \poly(\log \log n))$ and size $\poly(n)$.
\end{theorem}

In \cite{BCGQ}, Babai exhibited an algorithm that tests the equivalence of two linear codes of length $m$ in time $(2+o(1))^m$. He accomplished this by reducing to $\binom{m}{d}$ instances of \algprobm{Graph Isomorphism}-- precisely, isomorphism testing of $d \times (m-d)$ bipartite graphs with colored edges. We show that Babai's reduction is computable using an $\textsf{NC}$ circuit of depth $O(\log^2 m)$ and size $\poly(m)$. As $m \in O(\log n)$, we obtain that this reduction is \FOPLL-computable. 

\begin{proposition}[cf. {\cite[Theorem~7.1]{BCGQ}}] \label{prop:ReductionCodeEquivalence}
Equivalence of $d$-dimensional linear codes of length $m$, over any field, given by generators can be reduced to $\binom{m}{d}$ instances of isomorphism of $d \times (m-d)$ bipartite graphs with colored edges. In particular, if $m \in O(\log n)$, this reduction is \FOPLL-computable, assuming field operations at unit cost.
\end{proposition}

\begin{proof}
We follow the strategy as in the proof of \cite[Theorem~7.1]{BCGQ}. Let $A, B \in \text{Mat}_{d \times m}(\mathbb{F})$ be matrices of rank $d$. We need to find the set of pairs $(T, P) \in \text{GL}_{d}(\mathbb{F}) \times \text{Sym}(m)$ (where we treat $P$ as a permutation matrix) such that $B = TAP$. Our first goal is to place $A$ in the form $[I_{d} | A_{1}]$. This step is  computable with an $\textsf{NC}$ circuit of depth $O(\log^2 m)$ and size $\poly(m)$-- hence, in \FOPLL-- using Gaussian elimination \cite{BorodinParallelMatrix, MulmuleyRank}.  So for the remainder of the proof, we will abuse notation and refer to the previous $A$ as $A = [I_{d} |A_{1}]$.

The remainder of the proof is identical as in \cite[Theorem~7.1]{BCGQ}, and these details provide the precise reduction to \algprobm{GI}. Thus, for completeness, we will recall these details here. At a multiplicative cost of $\binom{m}{d}$ to the circuit size, we may in parallel, consider all such $d$-element subsets of $[m]$. For clarity, fix such a subset $[d]^{\sigma} \subseteq [m]$ under a hypothetical $\sigma \in \text{Sym}(m)$ making $A$ equivalent to $B$. By applying a single permutation to the columns, we reduce this case to those $\sigma$ satisfying $[d]^{\sigma} = [d]$. Following the proof of \cite[Theorem~7.1]{BCGQ}, we will refer to such $\sigma$ as \emph{basic equivalences}. The corresponding permtuation matrix $P(\sigma)$ can be written as a block-diagonal matrix $\text{diag}(P_{0}, P_{1})$, where $P_{0} \in \text{Mat}_{d \times d}(\mathbb{F})$ and $P_{1} \in \text{Mat}_{(m-d) \times (m-d)}(\mathbb{F})$ are permutation matrices.

We reduce the search for basic equivalences to one instance of finding the isomorphisms of $d \times (n-d)$ bipartite graphs, where the edges are colored using the elements of $\mathbb{F}$.  Let $B = [B_0 | B_1]$, where $B_0 \in \text{Mat}_{d \times d}(\mathbb{F})$ consists of the first $d$ columns of $B$. For any pair $(T, P)$ such that $B = TAP$, we will necessarily have that $B_{0} = TP_{0}$. Thus, if $B_{0}$ is singular, then there are no basic equivalences.

Suppose now that $B_{0} \in \text{GL}_{d}(\mathbb{F})$. Then $B_{0}B = [I_{d} | B_{0}^{-1}B_{1}]$. Thus, without loss of generality, we may assume that $B = [I_{d} | B_{1}]$. As we already have that $A = [I_{d} | A_{1}]$, we are looking for basic equivalences between $A = [I_{d} | A_{1}]$ and $B = [I_{d} | B_{1}]$. Now suppose that there exists a pair $(T, P) \in \text{GL}_{d}(\mathbb{F}) \times \text{Sym}(m)$, where $P = P(\sigma) = \text{diag}(P_0, P_1)$ is the permutation matrix corresponding to a basic equivalence $\sigma$, and $P_0, P_1$ are permutation matrices. Then $TP_{0} = \text{Id}_{d}$, which provides that $T = P_{0}^{-1}$ and $B_{1} = TA_{1}P_{1}$. Thus, we are searching for permutations $\sigma_0, \sigma_1$ such that $B_{1} = P(\sigma_{0}) A_{1} P(\sigma_{1})$. This is precisely the problem of deciding isomorphism of bipartite graphs with $\mathbb{F}$-colored edges, given by their incidence matrices $A_1, B_1$. The result now follows.	
\end{proof}

We now prove Theorem~\ref{thm:LinearCodeEquivalence}.

\begin{proof}[Proof of Theorem~\ref{thm:LinearCodeEquivalence}]
By Proposition~\ref{prop:ReductionCodeEquivalence}, we can in \FOPLL, reduce to the case of testing isomorphism of $\binom{m}{d}$ graphs with $O(m)$ vertices. By Theorem~\ref{thm:LuksBoundedGenerator}, we can decide isomorphism of such graphs using an $\textsf{AC}$ circuit of depth $O((\log^2 n) \cdot \poly(\log \log n))$ and size $\poly(n)$. The result now follows.
\end{proof}

Qiao, Sarma, and Tang \cite{QST11} considered a slight generalization of \algprobm{Linear Code Equivalence}, which they called \algprobm{Generalized Code Equivalence}.

\begin{definition}[{\cite[Problem~2]{QST11}}]
The \algprobm{Generalized Code Equivalence} problem takes as input two $d \times m$ matrices $C_{1}, C_{2}$ over the field $\mathbb{F}$, and a permutation group $S \leq \text{Sym}(m)$. We then ask if there exists some $B \in \text{GL}_{d}(\mathbb{F})$ and a permutation matrix $P \in S$ such that $BC_{1}P = C_{2}$.
\end{definition}

\begin{corollary}[cf. {\cite[Corollary~5.2]{QST11}}] \label{cor:GeneralizedCodeEquivalence}
Let $m \in O(\log n)$, and let $C_1, C_2$ be two linear codes of length $m$ given by their generator matrices. Let $S \leq \text{Sym}(m)$. We can solve the \algprobm{Generalized Code Equivalence} problem between $C_1$ and $C_2$ using an $\textsf{AC}$ circuit of depth $O((\log^2 n) \cdot \poly(\log \log n))$ and size $\poly(n)$, as well as compute the coset of equivalences.
\end{corollary}

\begin{proof}
Qiao, Sarma, and Tang \cite[Corollary~5.2]{QST11} observed that \algprobm{Generalized Code Equivalence} may be solved by first solving \algprobm{Linear Code Equivalence}, and then performing \algprobm{Coset Intersection} with $S$. By Theorem~\ref{thm:LinearCodeEquivalence}, we can compute the coset of equivalences for \algprobm{Linear Code Equivalence} between $C_{1}$ and $C_{2}$ using an $\textsf{AC}$ circuit of depth $O((\log^2 n) \cdot \poly(\log \log n))$ and size $\poly(n)$. The \algprobm{Coset Intersection} procedure is $\FOPLL$-computable (Lemma~\ref{PermutationGroupsNC}\ref{CosetInt}). The result now follows.
\end{proof}

%%%%%%%%%%%%%%%%
\section{Isomorphism Testing of Coprime Extensions in Parallel}

Recall that $\mathcal{H}(\mathcal{A}, \mathcal{E})$ is the class of coprime extensions $H \ltimes N$ such that $N$ is Abelian and $H$ is elementary Abelian. In this section, we will establish the following.

\begin{theorem}[cf. {\cite{QST11}}] \label{thm:CoprimeExtensions}
Let $G_1, G_2$ be groups given by their multiplication tables. There exists a uniform $\textsf{AC}^{3}$ algorithm that can decide if $G_1, G_2 \in \mathcal{H}(\mathcal{A}, \mathcal{E})$; and if so, decides whether $G_1 \cong G_2$.
\end{theorem}

%%%%
\subsection{Additional Preliminaries}

We begin by recalling some additional preliminaries concerning coprime extensions. We note that coprime extensions are determined entirely by the isomorphism types of $N$, $H$ and their actions-- see Taunt's lemma (recalled as Lemma~\ref{lem:Semidirect}). We will now show how to compute a decomposition $G = H \ltimes N$, where $N$ is Abelian, $H$ is elementary Abelian, and $\text{gcd}(|H|, |N|) = 1$, if such a decomposition exists.

\begin{lemma} \label{lem:DecomposeCoprime}
Let $G$ be a group given by its multiplication table. We can, in $\textsf{FL} \cap \textsf{FOLL}$, decide if there exists an Abelian normal Hall subgroup $N \trianglelefteq G$ such that $G/N$ is elementary Abelian and $\text{gcd}(|N|, |G/N|) = 1$, as well as compute such an $N$, if one exists, and $G/N$ in $\textsf{FL} \cap \textsf{FOLL}$.
\end{lemma}

\begin{proof}
We first note that any such $N$ is the direct product of a subset of the Abelian normal Sylow subgroups corresponding to the prime divisors of $|N|$. Thus, it suffices to take $N$ as the direct product of each Abelian normal Sylow subgroup of $G$. We proceed as follows, to first construct $N$. For each prime $p$ dividing $|G|$, we may in $\textsf{FL} \cap \textsf{FOLL}$,  compute the set $S_{p}$ of elements in $G$ whose order is a power of $p$ \cite{BKLM}. We then test, in $\ACz$, whether $S_{p}$ satisfies the group axioms and whether it is Abelian. If both of these conditions are satisfied, then $S_{p}$ forms the unique Abelian normal Sylow $p$-subgroup of $G$. Thus, we can construct $N$ in $\textsf{FL} \cap \FOLL$.

We now turn to constructing $H$. While the Schur--Zassenhaus theorem (cf. \cite{Robinson1982}) guarantees that there exists $H \leq G$ such that $G = H \ltimes N$, we will not explicitly construct such an $H$. Instead, we will use the quotient group representation $G/N$. Note that in $\textsf{AC}^{0}$, we may test whether any two elements $g, h \in G$ belong to the same coset of $G/N$, by testing whether $gh^{-1} \in N$.

Next, we claim that for any two elements $g_{1}, g_{2}$ belonging to the same coset of $G/N$ and any $n \in N$, we have that $n^{g_{1}} = n^{g_{2}}$. As $g_{1}, g_{2}$ belong to the same coset of $G/N$, there exists some $n' \in N$ such that $g_{2} = g_{1}n'$. As $N$ is Abelian, we have that: 
\[
n^{g_{2}} = g_{1}n' \cdot n \cdot (n')^{-1} g_{1}^{-1} = n^{g_{1}}.
\]
Thus, the quotient group representation of $H \cong G/N$ fully specifies the action. The result now follows.
\end{proof}

\subsection{Representations of $\mathbb{Z}_{q}^{\ell}$ over $\mathbb{Z}_{p}$}

Let $p, q$ be distinct primes. We recall key facts regarding representations of $\mathbb{Z}_{q}^{\ell}$ over $\mathbb{Z}_{p}$, as well as additional background from \cite[Section~5.2]{QST11}. For $n \in \mathbb{N}$, the \emph{cyclotomic polynomial} $\Phi_{n}(x)$ is the unique irreducible polynomial that is a divisor of $x^{n}-1$, but for all $k < n$, is not a divisor of $x^{k}-1$. Suppose that $\Phi_{q}(x)$ factors as $g_{1}(x) \cdots g_{r}(x)$, where $g_{1}(x), \ldots, g_{r}(x) \in \mathbb{Z}_{p}[x]$ are monic polynomials of degree $d = (q-1)/r$. Note that $d = |\mathbb{Z}_{q}^{\times}|$. Let $M \in \text{GL}_{d}(\mathbb{Z}_{p})$ be the companion matrix of $g_{1}(x)$. For $v \in \mathbb{Z}_{q}^{\ell}$ with $v \neq 0$, define $v^{*} : \mathbb{Z}_{q}^{\ell} \to \mathbb{Z}_{q}$ by mapping $v^{*}(u) = (v,u)$ (the inner product of $v$ and $u$). %Qiao, Sarma, and Tang \cite[Footnote~4]{QST11} noted that any $d \times d$ matrix with characteristic polynomial $g_{1}(x)$ would suffice, and we could have equivalently chosen instead a matrix with characteristic polynomial $g_{i}(x)$ for any $i \in [r]$.

For $u \in \mathbb{Z}_{q}^{\ell}$, let $0 \leq h_{u} < q$ such that $h_{u} \equiv v^{*}(u) \pmod{q}$. Now define $f_{v} : \mathbb{Z}_{q}^{\ell} \to \text{GL}_{d}(\mathbb{Z}_{p})$ by sending $u \mapsto M^{h_{u}}$. We will abuse notation by writing $M^{v^{*}(u)}$ in place of $M^{h_{u}}$. Let $f_{0} : \mathbb{Z}_{q}^{\ell} \to \text{GL}_{d}(\mathbb{Z}_{p})$ be the trivial representation. Now $\{ f_{v} : v \in \mathbb{Z}_{q}^{\ell} \}$ forms the set of all irreducible representations of $\mathbb{Z}_{q}^{\ell}$ over $\mathbb{Z}_{p}$. Note that for distinct and non-zero $u, v \in \mathbb{Z}_{q}^{\ell}$, $f_{u}$ and $f_{v}$ might be equivalent. We recall the following characterization from Qiao, Sarma, and Tang \cite{QST11}.

\begin{lemma}[{\cite[Claim~1]{QST11}}] \label{lem:Claim1}
Let $u, v \in \mathbb{Z}_{q}^{\ell}$ be distinct and both non-zero. Let $f_{u}, f_{v}$ be the corresponding irreducible representations. We have that $f_{u}$ and $f_{v}$ are equivalent if and only if there exists some $s \in \mathbb{Z}_{q}$ such that $u = sv$, and $M^{s}$ and $M$ are conjugate (by abuse of notation, we associate $s$ with the least non-negative integer belonging to the equivalence class $s \in \mathbb{Z}_{q}$).
\end{lemma}

Let $\tau : \mathbb{Z}_{q}^{\ell} \to \text{GL}_{k}(\mathbb{Z}_{p})$. As $p, q$ are distinct primes, we have by Maschke's theorem that $\tau$ is completely reducible. Write $\tau = f_{v_{1}}^{k_{1}} \oplus \cdots \oplus f_{v_{t}}^{k_{t}}$, where each $v_{i} \in \mathbb{Z}_{q}^{\ell}$, and $k_{1} \geq k_{2} \geq \cdots \geq k_{t} \geq 1$. For a given multiplicity $w \in [k]$, define $L_{\tau}(w)$ to be the set of irreducible representations with multiplicity $w$ appearing in $\tau$. Now define $L_{\tau} := (L_{\tau}(w))_{w \in [k]}$. We have that $L_{\tau}$ determines $\tau$ up to equivalence. In order to deal with the concrete form of the representations, Qiao, Sarma, and Tang introduced the following.

\begin{definition}[{\cite[Definition~5.4]{QST11}}]
Let $\tau : \mathbb{Z}_{q}^{\ell} \to \text{GL}_{k}(\mathbb{Z}_{p})$, and let $w \in [k]$. Define $\mathcal{L}_{\tau}(w)$ to be a set of vectors such that for every irreducible representation $f \in L_{\tau}(w)$, there exists a unique $v \in \mathcal{L}_{\tau}(w)$ such that $f$ and $f_{v}$ are equivalent. Now define $\mathcal{L}_{\tau} := (\mathcal{L}_{\tau}(w))_{w \in [k]}$. We refer to such a tuple $\mathcal{L}_{\tau}$ as an \emph{indexing tuple} of $L_{\tau}$.
\end{definition} 

Qiao, Sarma, and Tang established that a representation $\tau : \mathbb{Z}_{q}^{\ell} \to \text{GL}_{k}(\mathbb{Z}_{p})$ has at most $2^{k}$ indexing tuples \cite{QST11}.

\begin{proposition}[{\cite[Claim~2]{QST11}}] \label{prop:QSTClaim2}
Let $\tau, \gamma : \mathbb{Z}_{q}^{\ell} \to \text{GL}_{k}(\mathbb{Z}_{p})$ be two representations. We have that $\tau$ and $\gamma$ are equivalent if and only if there exist indexing tuples of $\tau$ and $\gamma$, $\mathcal{L}_{\tau}$ and $\mathcal{L}_{\gamma}$, such that $\mathcal{L}_{\tau} = \mathcal{L}_{\gamma}$. 
\end{proposition}

Let $\varphi \in \text{GL}_{\ell}(\mathbb{Z}_{q})$. The representation of $f_{v}$ induced by $\varphi$ has the following form: $(f_{v} \circ \varphi)(u) = f_{v}(\varphi(u)) = M^{v^{*}(\varphi(u))} = M^{(\varphi^{T}(v))^{*}(u)} = f_{\varphi^{T}(v)}(u)$. Now for any two representations $g, h$ of an arbitrary group  and any $\varphi' \in \Aut(G)$, we have that $(g \oplus h) \circ \varphi' = (g \circ \varphi') \oplus (h \circ \varphi')$. Thus, if $\tau : \mathbb{Z}_{q}^{\ell} \to \text{GL}_{k}(\mathbb{Z}_{p})$ is a representation, then $\tau \circ \varphi = f_{\varphi^{T}(v_{1})}^{k_{1}} \oplus \cdots \oplus f_{\varphi^{T}(v_{t})}^{k_{t}}$. If $S \subseteq \mathbb{Z}_{q}^{\ell}$, then $S^{\varphi} = \{ \varphi^{T}(v) : v \in S\}$. Thus, for any $\tau : \mathbb{Z}_{q}^{\ell} \to \text{GL}_{k}(\mathbb{Z}_{p})$, we have that $\mathcal{L}_{\tau \circ \varphi} = \mathcal{L}_{\tau}^{\varphi} = (\mathcal{L}_{\tau}(w)^{\varphi})_{w \in [k]}$.

\begin{lemma} \label{lem:ListIndexingTuples}
Let $p, q$ be distinct primes, and let $G = \mathbb{Z}_{q}^{\ell} \ltimes_{\varphi} \mathbb{Z}_{p}^{k}$ be given by its multiplication table. We can list all indexing tuples of $\varphi$ in $\textsf{NC}^{3}$.
\end{lemma}

%{prop:DecomposeRepresentation}

\begin{proof}
We first factor the cyclotomic polynomial $\Phi_{q}(x)$ in $\textsf{NC}^{3}$ \cite{Berlekamp} (see \cite{Eberly}). Let $g_{1}(x), \ldots, g_{r}(x)$ be the irreducible factors of $\Phi_{q}(x)$. In $\textsf{AC}^{0}$, we may write down the companion matrix $M$ for $g_{1}(x)$. By Proposition~\ref{prop:DecomposeRepresentation}, we can compute the decomposition of $\varphi = \tau_{1}^{k_{1}} \oplus \cdots \oplus \tau_{t}^{k_{t}}$ into its irreducible components, in $\textsf{AC}^{1}$. In particular, each irreducible representation is given explicitly by identifying each group element in $\mathbb{Z}_{q}^{\ell}$ with a corresponding matrix. Furthermore, we may in $\textsf{FL}$, sort the $\tau_{i}$ to ensure that $k_{1} \geq k_{2} \geq \cdots \geq k_{t}$.

Now recall that for each irreducible representation $\tau$ in the decomposition of $\varphi$, there exists some $v \in \mathbb{Z}_{q}^{\ell}$ such that $\tau$ is equivalent to $f_{v}$. We will proceed in parallel for each non-zero $v \in \mathbb{Z}_{q}^{\ell}$, to explicitly construct $f_{v}$. For clarity, we will consider some fixed $v \in \mathbb{Z}_{q}^{\ell}$. For each $u \in \mathbb{Z}_{q}^{\ell}$, we may compute $v^{*}(u) = (v, u)$ (the inner product) in parallel as follows. We first evaluate $(v,u)$ in the integers. As this is a sum of $O(\log n)$-many $O(\log n)$-bit integers, we can compute this sum in $\textsf{AC}^{0}$ (see e.g., \cite{VollmerText}). As the resulting sum uses only $O(\log n)$-many bits, we may  then reduce the sum modulo $q$ in $\textsf{AC}^{0}$. Thus, we may evaluate $v^{*}(u)$ in $\textsf{AC}^{0}$. Thus, in $\textsf{AC}^{0}$, we may select a corresponding $h_{u}$ such that $h_{u} \equiv v^{*}(u) \pmod{q}$.

We now turn to computing $M^{h_{u}}$. We will use the repeated squaring technique to perform $O(\log n)$ matrix multiplications. Observe first that if $A, B \in M_{d \times d}(\mathbb{Z}_{p})$, then $(AB)_{ij} = \sum_{h=1}^{d} A_{ih}B_{hj} \pmod{p}$. We may evaluate this sum in $\textsf{AC}^{0}$, by the discussion in the preceding paragraph. Thus, we may compute $M^{h_{u}}$ in $\textsf{AC}^{1}$. It follows that we can compute the map $f_{v}(u) = M^{h_{u}}$, for all $u \in \mathbb{Z}_{q}^{\ell}$, in $\textsf{AC}^{1}$.

Now for each $\tau_{i}$ ($i \in [t]$), we identify the set $S_{i}$ of $v \in \mathbb{Z}_{q}^{\ell}$ corresponding to the $f_{v}$ that are equivalent to $\tau_{i}$. By Lemma~\ref{lem:CompareCharacters}, this step is $\textsf{AC}^{0}$-computable. 

Now we may construct a single indexing set $\mathcal{L}_{\varphi}$ in the following manner. In parallel, for each $w \in [k]$ and each $i \in [t]$ such that $k_{i} = w$, we select a single vector $v \in S_{i}$ to include in $\mathcal{L}_{\tau}(w)$. This step is $\textsf{AC}^{0}$-computable. Qiao, Sarma, and Tang \cite{QST11} previously established that the number of indexing sequences for $\varphi$ is at most $2^k$. As $k \in O(\log n)$, there are $\poly(n)$ such indexing sequences. It follows that once the $S_{i}$ sets have been constructed, we may write down all such indexing sequences in $\ACz$.

Our algorithm is thus $\textsf{NC}^{3}$-computable, as desired. The result now follows.
\end{proof}

%%%%%
\subsection{Isomorphism Testing of $\mathcal{H}(\mathcal{E}, \mathcal{E})$} 

In this section, we will establish the following.

\begin{theorem}[cf. {\cite[Theorem~1.3]{QST11}}] \label{thm:ElemAbelian}
Given groups $G_{1}, G_{2}$ by their multiplication tables, there exists a uniform $\textsf{AC}^{3}$ algorithm that decides if $G_{1}, G_{2} \in \mathcal{H}(\mathcal{E}, \mathcal{E})$; and if so, decides if $G_{1} \cong G_{2}$. 
\end{theorem}

\begin{proof}
We first use Lemma~\ref{lem:DecomposeCoprime} to, in $\textsf{FL}$, find decompose $G_{i} = H_{i} \ltimes_{\tau_{i}} N_{i}$ ($i = 1, 2$). We may then, in $\textsf{L}$, decide whether: (i) the $N_{i}$ and $H_{i}$ are elementary Abelian groups of coprime order, (ii) $N_{1} \cong N_{2}$, and (iii) $H_{1} \cong H_{2}$ \cite{BKLM}. Suppose that $N_{1} \cong N_{2} \cong \mathbb{Z}_{p}^{k}$ and $H_{1} \cong H_{2} \cong \mathbb{Z}_{q}^{\ell}$. Now $\tau_{1}, \tau_{2} : \mathbb{Z}_{q}^{\ell} \to \text{GL}_{k}(\mathbb{Z}_{p})$ are representations. It remains to decide whether $\tau_{1}$ and $\tau_{2}$ are equivalent.

By Proposition~\ref{prop:DecomposeRepresentation}, we may in $\textsf{AC}^{1}$, decompose $\tau_{1}, \tau_{2}$ into their irreducible components and group them by equivalence type. Write $\tau_{1} = f_{v_{1}}^{k_{1}} \oplus \cdots \oplus f_{v_{t}}^{k_{t}}$ and $\tau_{2} = f_{u_{1}}^{\ell_{1}} \oplus \cdots \oplus f_{u_{t'}}^{\ell_{t'}}$. Now in $\textsf{NC}^{3}$ (using Lemma~\ref{lem:ListIndexingTuples}), we may obtain indexing tuples $\mathcal{L}_{\tau_{1}}, \mathcal{L}_{\tau_{2}}$. We may check in $\textsf{L}$ that $t = t'$; as well as that for all $w \in [k]$, whether $|\mathcal{L}_{\tau_{1}}(w)| = |\mathcal{L}_{\tau_{2}}(w)|$. If these conditions are not satisfied, then $\tau_{1}$ and $\tau_{2}$ are not equivalent; in which case, $G_{1}$ and $G_{2}$ are not isomorphic.

We now consider our fixed indexing tuple $\mathcal{L}_{\tau_{1}}$ for $\tau_{1}$. By Proposition~\ref{prop:QSTClaim2}, it suffices to decide if there exists an indexing tuple $\mathcal{L}_{\tau_{2}}$ of $\tau_{2}$ and $\varphi \in \text{GL}_{\ell}(\mathbb{Z}_{p})$ such that $\mathcal{L}_{\tau_{1}}^{\varphi} = \mathcal{L}_{\tau_{2}}$. In parallel, we will consider all indexing tuples of $\tau_{2}$. Note that we can list all such indexing tuples of $\tau_{2}$ in $\textsf{NC}^{3}$ (using Lemma~\ref{lem:ListIndexingTuples}). For clarity, fix such an indexing tuple $\mathcal{L}_{\tau_{2}}$. Qiao, Sarma, and Tang \cite[Proposition~5]{QST11} showed that deciding whether such a $\varphi$ exists reduces to \algprobm{Generalized Code Equivalence} for a code of length $m \in O(\log n)$ over $\prod_{w=1}^{k} \text{Sym}(|\mathcal{L}_{\tau_{1}}(w)|)$. As the symmetric group is $2$-generated, we can write down standard generators for $\text{Sym}(|\mathcal{L}_{\tau_{1}}(w)|)$ in $\ACz$. By Corollary~\ref{cor:GeneralizedCodeEquivalence}, we can solve this instance of \algprobm{Generalized Code Equivalence} using an $\textsf{AC}$ circuit of depth $O((\log^{2} n) \cdot \poly(\log \log n))$ and size $\poly(n)$. 

In total, our algorithm is $\textsf{AC}^{3}$-computable. The result now follows.
\end{proof}

As a consequence, we obtain the following corollary.

\begin{corollary}[cf. {\cite[Section~5.3]{QST11}}] \label{cor:ElemAbelian}
There exists a uniform $\textsf{AC}^{3}$ algorithm that, given groups $G_1, G_2$ by their multiplication tables, decides if $G_1, G_2 \in \mathcal{H}(\prod \mathcal{E}, \mathcal{E})$; and if so, decides if $G_1 \cong G_2$.
\end{corollary}

\begin{proof}
We first use Lemma~\ref{lem:DecomposeCoprime} write $G_{i} = H_{i} \ltimes_{\varphi_i} N_{i}$, where $N_i$ is the direct product of the Abelian normal Sylow subgroups of $G_i$. Now in $\textsf{L}$, we may compute the Abelian normal Sylow subgroups of $N_i$, and test whether each is elementary Abelian \cite{BKLM}. Similarly, in $\textsf{L}$, we test whether $H_i$ is elementary Abelian. Finally, in $\textsf{L}$, we may test whether $N_{1} \cong N_{2}$ and $H_{1} \cong H_{2}$ \cite{BKLM}. Otherwise, we reject.

For $i \in [2]$, let $A_{i,p} \leq N_{i}$ be the Sylow $p$-subgroup of $N_{i}$ (and hence, $G_{i}$). Let $\varphi_{i, p} : H_{i} \to \text{GL}(A_{i,p})$ be the projection of $\varphi_{i}$ onto $A_{i,p}$. Let $G_{i,p} = H_{i} \ltimes \varphi_{i,p} A_{i,p}$. Qiao, Sarma, and Tang \cite[Section~5.3]{QST11} established that $G_{1} \cong G_{2}$ if and only if, for all primes $p$ dividing $|N_{1}| = |N_{2}|$, $G_{1,p}$ and $G_{2,p}$ are isomorphic. We have shown that, in $\textsf{FL}$, we can construct each $G_{1,p}, G_{2,p}$. We now apply Theorem~\ref{thm:ElemAbelian} in parallel for each prime $p$ dividing $|N_{1}| = |N_{2}|$ to decide isomorphism between $G_{1,p}$ and $G_{2,p}$ in $\textsf{AC}^{3}$.

In total, our algorithm is $\textsf{AC}^{3}$-computable. The result now follows.
\end{proof}

%%%
\subsection{Le Gall's Technique in Parallel}

In this section, we will establish the following. 

\begin{proposition}[cf. {\cite{Gal09, BQ}}] \label{prop:LeGall}
Let $A$ be an Abelian $p$-group given by its multiplication table. Let $S = (g_{1}, \ldots, g_{s})$ be a basis for $A$, where $\mathbb{Z}_{p^{e_{i}}} \cong \langle g_{i} \rangle$. Let $\varphi_{1}, \varphi_{2} \in \Aut(A)$ be given by matrices in $R(A)$ with respect to $S$ (see Definition~\ref{def:Ranum}). Suppose that $|\varphi_{1}| = |\varphi_{2}|$. If $p$ does not divide $|\varphi_{1}| = |\varphi_{2}|$, then there exists an $\textsf{FL} \cap \textsf{FOLL}$ computable map $\Psi_{p} : \Aut(A) \to \text{GL}_{s}(\mathbb{Z}_{p})$ such that $\varphi_{1}$ and $\varphi_{2}$ are conjugate if and only if $\Psi_{p}(\varphi_{1})$ and $\Psi_{p}(\varphi_{2})$ are conjugate.
\end{proposition}

Following the strategy of \cite{Gal09, QST11, BQ}, we will later use Proposition~\ref{prop:LeGall} to reduce isomorphism testing of $\mathcal{H}(\mathcal{A}, \mathcal{E})$ to the case of $\mathcal{H}(\mathcal{E}, \mathcal{E})$. Le Gall \cite{Gal09} established Proposition~\ref{prop:LeGall} in the case when the complement was cyclic. He also showed that this reduction is $\textsf{NC}$-computable relative to $|A|$. Babai and Qiao \cite{BQ} subsequently showed that Le Gall's reduction holds for arbitrary complements.

Our goal in this section is to carefully analyze Le Gall's result, to show that $\Psi_{p}$ is $\textsf{FL} \cap \textsf{FOLL}$ computable. This is a critical step in order to establish the $\textsf{AC}^{3}$ bounds in Theorem~\ref{thm:CoprimeExtensions}. 

We first recall key background concerning the matrix characterization of the automorphism groups of Abelian groups, from Ranum \cite{Ranum} and Le Gall \cite[Section~4.2]{Gal09}. Throughout this section, we will assume that $A$ is an Abelian $p$-group. Suppose that $A \cong \mathbb{Z}_{p^{e_{1}}} \times \cdots \times \mathbb{Z}_{p^{e_{s}}}$, for some positive integers $s$ and $e_{1} \leq e_{2} \leq \cdots \leq e_{s}$. Let $(g_{1}, \ldots, g_{s})$ be a basis for $A$, where $\mathbb{Z}_{p^{e_{i}}} = \langle g_{i} \rangle$.

\begin{definition}[{\cite{Ranum}}] \label{def:Ranum}
Define $M(A)$ as the following subset of $s \times s$ matrices over $\mathbb{Z}$.
\[
M(A) = \{ (u_{ij}) \in \mathbb{Z}^{s \times s} \mid 0 \leq u_{ij} < p^{e_{i}} \text{ and } p^{e_{i} - e_{\min(i,j)}} \text{ divides } u_{ij} \text{ for all } i, j \in [s] \}.
\]
Given $U, U' \in M(A)$, define the multiplication operation $*$ as follows: $U * U'$ is the integer matrix $W$ of size $s \times s$ such that $w_{ij} = (\sum_{k=1}^{s} U_{ij}U'_{kj}) \pmod{p^{e_{i}}}$, for all $i, j \in [s]$. That is, after computing the usual matrix multiplication $UU'$, we reduce each entry modulo $p^{e_{i}}$, where $i$ is the row entry.

Let $R(A) = \{ U \in M(A) : \text{det}(U) \not \equiv 0 \pmod{p} \}$.
\end{definition}

\noindent Ranum established that if $A$ is an Abelian $p$-group, then $(R(A), *) \cong \Aut(A)$ \cite{Ranum}. We now recall key details regarding the structure of $R(A)$. Write $A = H_{1} \times \cdots \times H_{t}$, where $H_{i} = \mathbb{Z}_{p^{f_{i}}}^{k_{i}}$, $f_{1} < f_{2} < \cdots < f_{t}$, and each $k_{i} > 0$. So $f_{i}$ is the $i$th smallest element in the sequence $(e_1, \ldots, e_s)$, and $k_i$ is the multiplicity of $f_i$. Observe that $k_{1} + \cdots + k_{t} = s$. 

Let $U \in M(A)$. We define $t$ blocks $D_{1}(U), \ldots, D_{t}(U)$ as follows. $D_{i}(U)$ is the $k_{i} \times k_{i}$ matrix obtained by selecting the rows and columns with indices from $(k_{1} + \cdots + k_{i-1} + 1)$ to $(k_{1} + \cdots + k_{i-1} + k_{i})$. Observe that $D_{i}(U)$ lies on the diagonal of $U$. For any matrix $U \in M(A)$ and any $i \in [t]$, define $[U]_{i}$ as the matrix obtained by reducing the entries of $D_{i}(U)$ modulo $p$. Thus, $[U_{i}] \in \text{GL}_{k_{i}}(\mathbb{Z}_{p})$. For each $\ell \in [t]$, define the following subset of $M(H_\ell)$:
\[
K_{i}(A) = \{ (u_{ij}) \in M(H_\ell) \mid p \text{ divides } (u_{ij} - \delta_{ij}) \text{ for all } i, j \in [k_{\ell}]\}.
\]
Here, $\delta_{ij}$ is the Kronecker delta function. That is, $\delta_{ij} = 1$ if $i = j$, and $\delta_{ij} = 0$ otherwise. Thus, each diagonal entry of a matrix in $K_{i}(A)$ is of the form $1 + p\lambda_{ii}$, and each non-diagonal entry is of the form $p\lambda_{ij}$. We now recall the following definition from Le Gall.

\begin{definition}[{\cite{Gal09}}] \label{def:Def4.2}
Let $N(A)$ be the following subset of $M(A)$:
\[
N(A) := \{ U \in M(A) : D_{i}(U) \in K_{i}(A) \text{ for each } i \in [t] \}.
\]
We will also consider the following subgroup of $\text{GL}_{s}(\mathbb{Z}_{p})$.
\[
V(A) := \{ M \in \text{GL}_{s}(\mathbb{Z}_p) : V = \text{diag}(V_{1}, \ldots, V_{t}) \text{ with } V_{i} \in \text{GL}_{k_{i}}(\mathbb{Z}_{p}) \text{ for each } i \in [t] \}.
\]

Let $\Psi_{p} : R(A) \to V(A)$ be defined by $\Psi_{p}(U) = \text{diag}([U]_1, \ldots, [U]_t)$. That is, the diagonal blocks of $U$ are reduced modulo $p$, and the remaining entries are set to $0$.
\end{definition}

Le Gall established that $\Psi_{p}$ is a group homomorphism. We are now ready to prove Proposition~\ref{prop:LeGall}.

\begin{proof}[Proof of Proposition~\ref{prop:LeGall}]
We have from Le Gall \cite[Proposition~4.2]{Gal09} and Babai--Qiao \cite{BQ} that $\varphi_{1}, \varphi_{2}$ are conjugate if and only if $\Psi_{p}(\varphi_{1})$ and $\Psi_{p}(\varphi_{2})$ are conjugate. It remains to show that $\Psi_{p}$ is $\textsf{FL} \cap \textsf{FOLL}$-computable.

We first compute $|g_{i}|$ in $\textsf{FL} \cap \textsf{FOLL}$ \cite{BKLM}. Thus, in $\textsf{FL} \cap \textsf{FOLL}$, we can compute $t \leq s$ and $k_{1}, \ldots, k_{t}$ such that $A \cong \mathbb{Z}_{p^{f_{1}}}^{k_{1}} \times \cdots \times \mathbb{Z}_{p^{f_{t}}}^{k_{t}}$. Furthermore, in $\textsf{FL} \cap \textsf{FOLL}$, we can compute for each $i \in [t]$, the indices $(k_{1} + \cdots + k_{i-1} + 1)$ to $(k_{1} + \cdots + k_{i-1} + k_{i})$ corresponding to $D_{i}(\varphi_{1})$ (resp., $D_{i}(\varphi_{2})$). In particular, this allows us to explicitly write down each $D_{i}(\varphi_{1})$ (resp., $D_{i}(\varphi_{2})$). Now as each entry in $\varphi_{1}, \varphi_{2}$ is represented using $O(\log |A|)$ many bits, we may in $\ACz$ reduce the entries of each $D_{i}(\varphi_{1})$ (resp., each $D_{i}(\varphi_{2})$) modulo $p$, and then set the remaining entries of the original matrix to $0$. In total, $\Psi_{p}$ is $\textsf{FL} \cap \textsf{FOLL}$ computable. The result now follows.
\end{proof}

%%%%
\subsection{Isomorphism Testing of $\mathcal{H}(\mathcal{A}, \mathcal{E})$ in Parallel}

In this section, we prove Theorem~\ref{thm:CoprimeExtensions}.

\begin{proof}[Proof of Theorem~\ref{thm:CoprimeExtensions}]
We first use Lemma~\ref{lem:DecomposeCoprime} write $G_{i} = H_{i} \ltimes_{\varphi_i} N_{i}$, where $N_i$ is the direct product of the Abelian normal Sylow subgroups of $G_i$. Now in $\textsf{FL}$, we may compute the Abelian normal Sylow subgroups of $N_i$. Similarly, in $\textsf{L}$, we test whether $H_i$ is elementary Abelian. Finally, in $\textsf{L}$, we may test whether $N_{1} \cong N_{2}$ and $H_{1} \cong H_{2}$ \cite{BKLM,ChattopadhyayToranWagner,CGLWISSAC}. Otherwise, we reject.

For $i \in [2]$, let $A_{i,p} \leq N_{i}$ be the Sylow $p$-subgroup of $N_{i}$ (and hence, $G_{i}$). Let $\varphi_{i, p} : H_{i} \to \Aut(A_{i,p})$ be the projection of $\varphi_{i}$ to $A_{i,p}$. Let $G_{i,p} = H_{i} \ltimes \varphi_{i,p} A_{i,p}$. Qiao, Sarma, and Tang \cite[Section~5.3]{QST11} established that $G_{1} \cong G_{2}$ if and only if, for all primes $p$ dividing $|N_{1}| = |N_{2}|$, $G_{1,p} := H_{1} \ltimes_{\varphi_{1,p}} A_{1,p}$ and $G_{2,p} := H_{2} \ltimes_{\varphi_{2,p}} A_{2,p}$ are isomorphic. We have shown that, in $\textsf{FL}$, we can construct each $G_{1,p}, G_{2,p}$.

By \cite[Proposition~4.5]{JLVWQuasigroups}, we may in $\textsf{AC}^{1}$ compute a basis $\beta$ for $A_{1,p}$, as well as $A_{2,p}$, and thus fix an isomorphism $\phi_{p} : A_{1,p} \cong A_{2,p}$. Let $\Psi_{p}$ be as defined in Definition~\ref{def:Def4.2}. By Proposition~\ref{prop:LeGall} (applied with $\beta$), we have that $\varphi_{1,p}, \varphi_{2,p}$ are conjugate if and only if $\Psi_{p}(\varphi_{1,p})$ and $\Psi_{p}(\varphi_{2,p})$ are conjugate. Furthermore, $\Psi_{p}$ is $\textsf{FL} \cap \textsf{FOLL}$ computable. Thus, given $G_{1,p}$ and $G_{2,p}$, we may in $\textsf{AC}^{1}$ write down $G_{i,p}^{\prime} := H_{1} \ltimes_{\Psi_{p}(\varphi_{i,p})} \mathbb{Z}_{p}^{s_{p}}$ ($i = 1, 2$). As $H_{1} \cong H_{2}$ and $A_{1,p} \cong A_{2,p}$, we have that the following are equivalent: 
\begin{itemize}
\item $G_{1,p} \cong G_{2,p}$,
\item $G_{1,p}' \cong G_{2,p}'$, and
\item $\Psi_{p}(\varphi_{1})$ and $\Psi_{p}(\varphi_{2})$ are conjugate.
\end{itemize}

Thus, it suffices to test for isomorphism between $G_{1,p}' \cong G_{2,p}'$, which is $\textsf{AC}^{3}$-computable using Theorem~\ref{thm:ElemAbelian}. The result now follows.
\end{proof}

%%%%%%%

\section{Isomorphism Testing of Central-Radical Groups in Parallel: When Enumerating $\Aut(Q)$ is Allowed} \label{sec:Section6}

In this section, we will establish the following.

\begin{theorem}[cf. {\cite[Theorem~6.1]{GQCoho}}] \label{thm:GQCohoA}
Let $\mathcal{S}$ be a logspace-computable characteristic subgroup function. Fix functions $d(n) \in \Omega(\log^2 n)$ and $s(n)$. Let $G_1, G_2$ be two groups of order $n$ given by their multiplication tables, and suppose that (i) $\mathcal{S}(G_1) \leq Z(G_1)$ and (ii) $\Aut(G_1/\mathcal{S}(G_1))$ can be listed using an $\mathsf{AC}$ circuit of depth $d(n)$ and size $s(n)$. Then the following hold:
\begin{enumerate}[label=(\alph*)]
\item We can decide isomorphism between $G_1$ and $G_2$ using an $\mathsf{AC}$ circuit of depth $\max\{d(n), \log^{3}(n) \}$ and size $s(n) \cdot n^{O(1)}$. 

\item If furthermore, $\mathcal{S}(G_1)$ is elementary Abelian, then the coset of isomorphisms can be found using an $\textsf{AC}$ circuit of depth $d(n)$ and size $s(n) \cdot n^{O(1)}$.
\end{enumerate}
\end{theorem}

\begin{remark}
The key reason for the difference in circuit depth, between Theorem~\ref{thm:GQCohoA}(a) and (b), is due to the difference in parallel complexity for solving systems of linear equations. Over arbitrary Abelian groups, solving a system of linear equations is $\mathsf{NC}^{3}$-computable  \cite{McKenzieCook, MulmuleyRank}. In the case of elementary Abelian groups, solving a system of linear equations is $\textsf{NC}^{2}$-computable (see e.g., \cite{BorodinParallelMatrix, MulmuleyRank, BDHMParallelLinearAlgebra}).
\end{remark}

We will now record a couple corollaries of Theorem~\ref{thm:GQCohoA}. Grochow and Levet \cite{GrochowLevetWL} previously established that $\Aut(G/\rad(G))$ can be listed using an $\textsf{SAC}$ (and hence, $\textsf{AC}$) circuit of depth $O(\log n)$ and size $n^{O(\log \log n)}$. Subsequently, Grochow, Johnson, and Levet \cite{GrochowJohnsonLevet} showed that this could be accomplished using a $\mathsf{AC}$ circuit of depth $O(\log \log n)$ and size $n^{O(\log \log n)}$. We thus obtain the following corollary.

\begin{corollary}[cf. {\cite[Theorem~A/Corollary~6.2]{GQCoho}}] \label{cor:6.2}
Isomorphism of central-radical groups, given by their multiplication tables, can be decided by $\textsf{AC}$ circuits of depth $O(\log^3 n)$ and size $n^{O(\log \log n)}$. If furthermore the center is elementary Abelian, then we can decide isomorphism and compute the coset of isomorphisms using $\textsf{AC}$ circuits of depth $O(\log^2 n)$ and size $n^{O(\log \log n)}$.	
\end{corollary}

Babai, Codenotti, Grochow, and Qiao \cite{BCGQ} previously established that if $G/\rad(G)$ has $O(\log n / \log \log n)$ minimal normal subgroups, then $\Aut(G/\rad(G))$ has size $\poly(|G/\rad(G)|)$. Grochow and Levet \cite{GrochowLevetWL} established that in this case, we can list $\Aut(G/\rad(G))$ in logspace. Thus, we obtain the following corollary.

\begin{corollary}[cf. {\cite[Corollary~6.3]{GQCoho}}]
Let $G$ and $H$ be central-radical groups of order $n$, given by their multiplication tables. If $G/\rad(G)$ has $O(\log n / \log \log n)$ minimal normal subgroups, then isomorphism between $G$ and $H$ can be decided in $\textsf{AC}^{3}$. If furthermore, the center is elementary Abelian, then we can decide isomorphism and compute the coset of isomorphisms in $\textsf{AC}^{2}$.
\end{corollary}

\subsection{Additional Preliminaries}

We recall from Grochow and Qiao \cite{GQCoho} details on how to work with $2$-cohomology classes algorithmically.  First, as the action is trivial in central extensions, we will drop it from $Z^{2}(Q,A), B^{2}(Q,A)$, and $H^{2}(Q,A)$. Write $A := \mathcal{S}(G)$, and let $A = \prod_{i=1}^{k} \mathbb{Z}/p_{i}^{\mu_{i}}\mathbb{Z}$ be the  decomposition of $A$ into cyclic subgroups.%-- note that we may compute this decomposition in $\mathsf{AC}^{1}$ \cite{JLVWQuasigroups}, together with cyclic generators $e_{i}$ of $\mathbb{Z}/p_{i}^{\mu_{i}}\mathbb{Z}$.

By choosing an arbitrary section $s$, we get a cocycle $f : Q \times Q \to A$. We may view $f$ as a $k \times |Q|^{2}$-size integer matrix, which we denote $M_{f}$. The rows are indexed by $[k]$, and the columns are indexed by $Q \times Q$. For $i \in [k]$ and $(q, q') \in Q \times Q$, the entry $M_{f}[i, (q, q')]$ is the $i$th coordinate of $f(q, q')$ relative to the basis $(e_1, \ldots, e_k)$ modulo $p_{i}^{\mu_{i}}$. 

Under this identification, the set $C^{2}(Q, A)$ is the set of all such matrices. Now $Z^{2}(Q,A)$ is a subgroup of $C^{2}(Q,A)$, under matrix addition. Similarly, $B^{2}(Q,A)$ is a subgroup of $Z^{2}(Q,A)$ under matrix addition. We will use $U_{i}$ to denote the subgroup of $C^{2}(Q,A)$ consisting of matrices whose only non-zero entries are in the $i$th row (so $U_{i} \cong (\mathbb{Z}/p_{i}^{\mu_{i}}\mathbb{Z})^{|Q|^2}$). $\Aut(A)$ acts on $C^{2}(Q,A)$ by left-multiplication, and $\Aut(Q)$ acts on $C^{2}(Q,A)$ by permuting the columns according to the diagonal action of $\Aut(Q)$. Note that the actions of $\Aut(A)$ and $\Aut(Q)$ commute. 

\begin{proposition}[cf. {\cite[Proposition~6.8]{GQCoho}}] \label{prop:GQ6.8}
For any $\mu \geq 1$, a $\mathbb{Z}$-basis of $B^{2}(Q, \mathbb{Z}/p^{\mu}\mathbb{Z})$ can be computed in $\ACz$. Furthermore, a $\mathbb{Z}_{p}$-basis of $B^{2}(Q, \mathbb{Z}_{p})$ can be computed in the same bound.
\end{proposition}

\begin{proof}
We will parallelize \cite[Proposition~6.8]{GQCoho}. Let $A = \mathbb{Z}/p^{\mu}\mathbb{Z}$. For $q \in Q$, $q \neq \text{id}$, let $u_{q} : Q \to \mathbb{Z}$ be given by $u_{q}(q') = \delta(q, q')$, where $\delta$ is the Kronecker delta function. Let $f_{q} : Q \times Q \to A$ be the $2$-coboundary based on $u_{q}$. The set $V := \{ f_{q} : q \in Q \}$ then forms a basis for $B^{2}(Q, \mathbb{Z}/p^{\mu}\mathbb{Z})$. There are $|Q|$ basis elements, each of which is constructed by computing its $|Q|^{2}$ values. Each such value can be computed by a constant number of additions in $\mathbb{Z}/p^{\mu}\mathbb{Z}$ and a single  product in $Q$. Both of these steps are $\ACz$-computable. 

Note that if $\mu = 1$, then $V$ is a $\mathbb{Z}_{p}$-basis for $B^{2}(Q, \mathbb{Z}_{p})$. The result now follows.
\end{proof}

Now for a $2$-cochain $f \in C^{2}(Q, A)$ with corresponding $k \times |Q|^2$ matrix $M$, let $R_{i} \leq (\mathbb{Z}/p_{i}^{\mu_{i}}\mathbb{Z})^{|Q|^2}$ be the subgroup generated by the $i$th row of $M$. For $\mu < \mu_{i}$, let $R_{i}^{(\mu)}$ denote the subgroup of $(\mathbb{Z}/p_{i}^{\mu_{i}}\mathbb{Z})^{|Q|^2}$ that is given by taking $R_{i}$ modulo $p^{\mu}$. For $\mu > \mu_{i}$, let $R_{i}^{(\mu)}$ denote the subgroup of $(\mathbb{Z}/p_{i}^{\mu_{i}}\mathbb{Z})^{|Q|^2}$ that is given by multiplying every element of $R_{i}$ by $p^{\mu-\mu_{i}}$. For any prime $q \neq p$, let $R_{i}^{(q,\mu)}$ denote the trivial subgroup. If $q = p$, let $R_{i}^{(q, \mu)} := R_{i}^{(\mu)}$. Let $R^{(p, \mu)} = \langle R_{1}^{(p, \mu)}, \ldots, R_{k}^{(p, \mu)} \rangle$. 

\begin{proposition}[{\cite[Proposition~6.9]{GQCoho}}] \label{prop:GQ6.9}
Let $A = \prod_{i=1}^{k} \mathbb{Z}/p_{i}^{\mu_{i}}\mathbb{Z}$ be an Abelian group (the $p_i$ are primes, not necessarily distinct). Let $f_{1}, f_{2} \in C^{2}(Q, A)$. With the notation as above, there exists $\alpha \in \Aut(A)$ such that $f_{1}$ and $f_{2}^{\alpha}$ are cohomologous if and only if
\[
\langle R_{1}^{(p_{i}, \mu_{i})}, B^{2}(Q, \mathbb{Z}/p_{i}^{\mu_{i}}\mathbb{Z}) \rangle = \langle R_{2}^{(p_{i}, \mu_{i})}, B^{2}(Q, \mathbb{Z}/p_{i}^{\mu_{i}}\mathbb{Z}) \rangle,
\]
for each $i \in [k]$. Here, $\langle \cdot \rangle$ denotes the $\mathbb{Z}$-span ($=$group generated by).
\end{proposition}

\begin{proposition}[cf. {\cite[Proposition~6.10]{GQCoho}}] \label{prop:6.10}
For $A = \mathbb{Z}_{p}^{k}$ and a group $Q$, let $n = |A| \cdot |Q|$. In $\textsf{NC}^{2}$, we can compute an $\Aut(A)$-invariant complement $W$ of $B^{2}(Q,A)$ in $C^{2}(Q,A)$, as well as a $\mathbb{Z}_{p}$-linear projection $\pi : C^{2}(Q,A) \to W$ that commutes with every $\alpha \in \Aut(A)$.
\end{proposition}

\begin{proof}
We parallelize \cite[Proposition~6.10]{GQCoho}. By Proposition~\ref{prop:GQ6.8}, we can compute a basis $V_{0}$ of $B^{2}(Q, \mathbb{Z}_{p})$ in $\ACz$. Let $W_{0}$ be a linear complement of $V_{0}$ in $C^{2}(Q, \mathbb{Z}_{p})$ (which can be thought of as the space of row vectors of length $|Q|^2$). We can compute $W_{0}$ in $\textsf{NC}^{2}$ by solving an appropriate system of linear equations (see e.g., \cite{BorodinParallelMatrix, MulmuleyRank, BDHMParallelLinearAlgebra}). For each $i \in [k]$, let $R_{i}$ denote the subgroup of $C^{2}(Q,A)$ consisting of those matrices whose only non-zero entries are in the $i$th row, and let $B_{i}$ be the copy of $B^{2}(Q, \mathbb{Z}_{p})$ in $R_{i}$. Let $W_{i}$ be the copy of $W_{0}$ in $R_{i}$, and let $\pi_{i} : R_{i} \to W_{i}$ be the projection of $R_{i}$ to $W_{i}$ along $B_{i}$. (As Grochow and Qiao note, here we are identifying each $R_{i}$ with $C^{2}(Q, \mathbb{Z}_{p})$ as in \cite[Proposition~6.5]{GQCoho}). Now define $\pi : \oplus_{i=1}^{k} R_{i} \to \oplus_{i=1}^{k} W_{i}$ by $\pi(x_1, \ldots, x_k) = (\pi_{1}(x_1), \ldots, \pi_{k}(x_{k}))$ to be the projection along $\oplus_{i=1}^{k} B_{i}$. Note that $C^{2}(Q,A) = \oplus_{i=1}^{k} R_{i}$. Grochow and Qiao previously established that we can construct each $\pi_{i}$ ($i \in [k]$), and hence $\pi$, by solving an appropriate system of linear equations. We can compute such a solution in $\textsf{NC}^{2}$ (see e.g., \cite{BorodinParallelMatrix, MulmuleyRank, BDHMParallelLinearAlgebra}).

Grohow and Qiao previously established that $f, g \in Z^{2}(Q,A)$ are cohomologous if and only if $\pi(f) = \pi(g)$. Furthermore, Grochow and Qiao previously established that for any $\alpha \in \Aut(A)$, $\pi \alpha = \alpha \pi$. The result now follows.
\end{proof}

We will also need the following observation from Grochow and Qiao \cite{GQCoho}.

\begin{observation}[{\cite[Observation~5.2]{GQCoho}}] \label{obs:5.2}
Let $X$ be a subgroup of an extension $Y \xhookrightarrow{\iota} \overline{X} \overset{\pi}{\twoheadrightarrow} Z$. Suppose that  $\mathcal{Y} \subseteq \iota(Y)$ generates $\iota(Y) \cap X$, and $\mathcal{Z} \subseteq Z$ generates $\pi(X)$. Furthermore, suppose that for each $z \in \mathcal{Z}$, there exists $x_{z} \in X$ such that $\pi(x_{z}) = z$. Then $\mathcal{Y} \cup \{ x_{z} \in X : z \in \mathcal{Z} \}$ generates $X$.
\end{observation}

We will use Observation~\ref{obs:5.2} when $X, Y, Z$ are automorphism groups of other groups.

\noindent We are now ready to prove Theorem~\ref{thm:GQCohoA}(a).

\begin{proof}[Proof of Theorem~\ref{thm:GQCohoA}(a)]
We follow the strategy of \cite[Theorem~6.1]{GQCoho}. We begin by listing $\Aut(Q)$ using an $\mathsf{AC}$ circuit of size $s(n)$ and depth $d(n)$. Choose arbitrary sections of $G, H$, to get a $2$-cocycle $f_{1}$ of $G$ and a $2$-cocycle $f_{2}$ of $H$. By Lemma~\ref{lem:GQMainLemma}, it is sufficient and necessary to test whether there exist $(\alpha, \beta) \in \Aut(A) \times \Aut(Q)$ such that $f_{1}$ and $f_{2}^{(\alpha, \beta)}$ are cohomologous. 

For each $\beta \in \Aut(Q)$, we get $f_{2}' := f_{2}^{(\text{id}, \beta)}$. We first use Proposition~\ref{prop:GQ6.8} to, in $\ACz$, get a basis $V$ for $B^{2}(Q, A)$. Let $M_{1}$ be the matrix representation of $f_{1}$, and $M_{2}$ be the matrix representation for $f_{2}'$. By Proposition~\ref{prop:GQ6.9}, it suffices to determine whether the $\mathbb{Z}$-span of the rows in $M_{1}$ with $V$, is the same as the $\mathbb{Z}$-linear span of the rows of $M_{2}$ with $V$. Checking this condition reduces to solving a system of linear equations over the Abelian group $A$, which is $\mathsf{NC}^{3}$-computable  \cite{McKenzieCook, MulmuleyRank}. The result now follows.
\end{proof}

%%%%%
\subsection{Proof of Theorem~\ref{thm:GQCohoA}(b)}

In this section, we will prove Theorem~\ref{thm:GQCohoA}(b). 

\begin{proof}[Proof of Theorem~\ref{thm:GQCohoA}(b)]
We parallelize the proof from Grochow and Qiao \cite[Section~6.1.2]{GQCoho}. Here, we will assume that the solvable radical is central and elementary Abelian.
\begin{itemize}
\item \textbf{Deciding isomorphism.} This proof is slightly different than Theorem~\ref{thm:GQCohoA}(a), in that it provides the key structure for computing the full coset of isomorphisms. Instead of including a basis of $B^{2}(Q,A)$ among the rows of an extended matrix $M_{f_{j}}$, we use an $\Aut(A)$-invariant projection from $C^{2}(Q,A) \to C^{2}(Q,A)/B^{2}(Q,A)$ (Proposition~\ref{prop:6.10}). We now proceed with the details. 

We list $\Aut(Q)$ using an $\textsf{AC}$ circuit of depth $d(n)$ and size $s(n)$. For $i = 1, 2$, fix an arbitrary section of $Q$ in $G_{i}$ to obtain a $2$-cocycle $f_{i}$. We use Proposition~\ref{prop:6.10} to obtain the following in $\textsf{NC}^{2}$: 
\begin{itemize}
\item a complement $W$ of $B^{2}(Q,A)$ in $C^{2}(Q,A)$ such that $W$ is $\Aut(A)$-invariant; and 
\item a projection $\pi : C^{2}(Q,A) \to W$, such that for every $\alpha \in \Aut(A)$, $\pi \alpha = \alpha \pi$. 
\end{itemize}

By Lemma~\ref{lem:GQMainLemma}, it is necessary and sufficient to test whether there exists an $(\alpha, \beta) \in \Aut(A) \times \Aut(Q)$ such that $\pi(f_{1}) = \pi(f_{2}^{(\alpha, \beta)})$. As every $\alpha \in \Aut(A)$ commutes with both $\pi$ and every $\beta \in \Aut(Q)$, this condition is equivalent to $\pi(f_{1}) = (\pi(f_{2}^{(\text{id}, \beta)}))^{(\alpha, \text{id})}$. That is, we may leave $\alpha$ unspecified until the final step.

For each $\beta \in \Aut(Q)$, compute $f_{1}' = \pi(f_{1})$ and $f_{2}' = \pi(f_{2}^{(\text{id}, \beta)})$. Note that $f_{1}', f_{2}' \in W \leq C^{2}(Q, A)$. It remains to find an $\alpha \in \Aut(A)$ such that $f_{1}' = \alpha^{-1}(f_{2}')$. There exists such an $\alpha \in \Aut(A)$ if and only if the row spans of $M_{f_{1}'}$ and $M_{f_{2}'}$ are the same in $\mathbb{Z}_{p}^{|Q|^{2}}$. This latter task reduces to solving an appropriate system of linear equations, which is $\textsf{NC}^{2}$-computable (see e.g., \cite{BorodinParallelMatrix, MulmuleyRank, BDHMParallelLinearAlgebra}). Lemma~\ref{lem:GQMainLemma} implies that $G_{1} \cong G_{2}$ if and only if the preceding test succeeds for some $\beta \in \Aut(Q)$.

In total, we utilized an $\textsf{AC}$ circuit of depth $d(n)$ and size $s(n) \cdot \poly(n)$, as desired.

\item \textbf{Computing the coset of isomorphisms.} Note that $\Aut(G)$ maps into $\Aut(A) \times \Aut(Q)$. Let $\rho : \Aut(G) \to \Aut(Q) \times \Aut(Q)$ be this homomorphism. We will apply Observation~\ref{obs:5.2} twice: once to split the computation into computing generators for $\text{Im}(\rho) \leq \Aut(A) \times \Aut(Q)$ and $\text{Ker}(\rho)$; and then once to split the computation of $\text{Im}(\rho)$ into its projections onto $\Aut(A)$ and $\Aut(Q)$ respectively. Following the strategy of Grochow and Qiao, we will handle this latter step first.

Using the notation of Observation~\ref{obs:5.2}, we will take $Y = \Aut(A)$ and $Z = \Aut(Q)$. As we can enumerate $\Aut(Q)$, we may in fact take $\mathcal{Z}$ to be the entirety of the projection of $\Aut(G)$ onto $\Aut(Q)$. In order to use Observation~\ref{obs:5.2}, for each automorphism of $Q$ that admits a compatible automorphism of $A$, we need to construct one such automorphism of $A$. This is handled in $\textsf{NC}^{2}$ using parallel linear algebra, as described in the procedure above to decide isomorphism. 

We then also need the automorphisms of $A$ that preserve the given cohomology class $f : Q \times Q \to A$; that is, the automorphisms of $A$ that send $f$ to itself modulo $B^{2}(Q,A)$. Equivalently, we want the stabilizer of the matrix $M_{\pi(f)}$ in $\Aut(A) \cong \text{GL}_{k}(\mathbb{Z}_{p})$. As $\Aut(A)$ acts on each column of a $2$-cochain independently, the stabilizer is the same as the pointwise stabilizer of the \emph{set} of columns of $M_{\pi(f)}$. If the columns of $M_{\pi(f)}$ span $A$, then this stabilizer is trivial, and we are done. More generally, let $B \leq A$ be the space spanned by the columns of $M_{\pi(f)}$. In $\textsf{AC}^{1}$, we may construct a basis $\{ e_{1}, \ldots, e_{k}\}$ for $A$ such that $\{ e_{1}, \ldots, e_{\text{dim}(B)} \}$ restricts to a basis for $B$ \cite{JLVWQuasigroups}. In this basis, the stabilizer of $M_{\pi(f)}$ in $\Aut(A) \cong \text{GL}_{k}(\mathbb{Z}_{p})$ consists of all block matrices of the form
\[
\begin{pmatrix} 
\text{id} & * \\
0 & \nu
\end{pmatrix},
\]
where $\nu \in \text{GL}_{k-\text{dim}(B)}(\mathbb{Z}_{p})$ and $*$ indicates any $(\text{dim}(B)) \times (k-\text{dim}(B))$ matrix. A generating set for this subgroup is $\ACz$-computable, by writing down a standard set of generating matrices for $\text{GL}$ (e.g., elementary matrices). Observation~\ref{obs:5.2} then yields a generating set for $\text{Im}(\rho) \leq \Aut(A) \times \Aut(Q)$. 

Finally, we apply Observation~\ref{obs:5.2}, with $Y = \text{Ker}(\rho)$ and $Z = \text{Im}(\rho)$. We must do the following:
\begin{itemize}
\item For each such generator of $\mathcal{Z}$, we compute its lift in $\Aut(G)$.
\item Compute a set of generators for $\text{Ker}(\rho)$.
\end{itemize}

We will first compute a lift of $z$ in $\Aut(G)$. Fix a section $s : Q \to G$. We will proceed in parallel, for each generator of $\mathcal{Z}$. For clarity, fix a generator $z \in \mathcal{Z}$. Given $(\alpha, \beta) \in \text{Im}(\rho) \leq \Aut(A) \times \Aut(Q)$, every element of $G$ is uniquiely represented as $a \cdot s(q)$ for some $a \in A, q \in Q$. The corresponding automorphism of $G$ then acts by sending $a \cdot s(q) \mapsto \alpha(a) \cdot s(\beta(q))$. It is readily verified that this is indeed an automorphism of $G$.

Now $\text{Ker}(\rho)$ consists of those automorphisms of $G$ that fix $Q$ pointwise and fix $A$ pointwise. In particular, the automorphisms in $\text{Ker}(\rho)$ can only move elements around within their respective cosets modulo $A$. Thus, any automorphism in $\text{Ker}(\rho)$ is determined by a map $\delta : Q \to A$ such that the automorphism is given by sending $a \cdot s(q) \mapsto a \cdot \delta(q) \cdot s(q)$. Grochow and Qiao previously established that this map is indeed a homomorphism. These $|Q|^{2}$ equations form a homogenous system of linear equations in $|Q| \cdot k$ variables over $\mathbb{Z}_{p}$. We can solve such a system in $\textsf{NC}^{2}$ (see e.g., \cite{BorodinParallelMatrix, MulmuleyRank, BDHMParallelLinearAlgebra}), to recover generators for $\text{Ker}(\rho)$. Observation~\ref{obs:5.2} now yields generators for $\Aut(G)$.

In total, we utilized an $\textsf{AC}$ circuit of depth $d(n)$ and size $s(n) \cdot \poly(n)$. The result now follows. \qedhere
\end{itemize}
\end{proof}

%%%%%%%
\section{Isomorphism Testing of Central-Radical Groups in Parallel: When $\Aut(Q)$ is too big}

In this section, we will establish the following theorem.

\begin{theorem}[cf. {\cite[Theorem~C]{GQCoho}}] \label{thm:MainCohoSec7}
Let $G_{1}, G_{2}$ be groups given by their multiplication tables. Suppose that $G_1, G_2$ have central, elementary Abelian radicals. Then isomorphism can be decided, and the coset of isomorphisms found, in $\textsf{AC}^{3}$, if either:
\begin{enumerate}[label=(\alph*)]
\item $G_{1}/\rad(G_1)$ is a direct product of non-Abelian simple groups; or
\item $G_{1}/\rad(G_1)$ is a direct product of perfect groups, each of order $O(1)$.
\end{enumerate}
\end{theorem}

\subsection{Additional Preliminaries}

We will begin by recalling some additional preliminaries from \cite[Section~7]{GQCoho}. The following lemma from Grochow and Qiao establishes conditions in which the \emph{cohomology splits}, allowing us to restrict attention to the extensions corresponding to the individual direct factors of the quotient. Throughout this section, we will consider central extensions $A\xhookrightarrow{} G_{j} \twoheadrightarrow Q$ ($j = 1, 2$), with $A = Z(G_j)$, and $Q = \prod_{i=1}^{\ell} T_{i}$ with each $T_i$ perfect, centerless, and indecomposable.

\begin{lemma}[{\cite[Lemma~7.4]{GQCoho}}] \label{lem:7.4}
Given two central extensions $A\xhookrightarrow{} G_{j} \twoheadrightarrow Q$ ($j = 1, 2$), with $A = Z(G_j)$, and $Q = \prod_{i=1}^{\ell} T_{i}$ with each $T_i$ perfect, centerless, and indecomposable. Let $U_{j,i}$ be the preimage of $T_{i}$ under the natural projection map $G_{j} \to G_{j}/A$. The extensions $A\xhookrightarrow{} G_{j} \twoheadrightarrow Q$ ($j = 1, 2$) are equivalent if and only if for all $i \in [\ell]$, the extensions $A\xhookrightarrow{} U_{j,i} \twoheadrightarrow T_i$ ($j = 1, 2$) are equivalent.
\end{lemma}

\begin{lemma} \label{lem:DecomposeCentralRadical}
Let $G$ be a group of order $n$. We can decide if (i) $\rad(G) = Z(G)$ is elementary Abelian; and if so, (ii) compute the decomposition of $G/Z(G)$ into a direct product of either non-Abelian simple groups or $O(1)$-size perfect, indecomposable groups in $\mathsf{FL}$. Furthermore, we can group the direct factors by isomorphism type in $\mathsf{FL}$.
\end{lemma}

\begin{proof}
We first compute $Z(G)$ in $\ACz$ by identifying those group elements $g$ which compute with every $x \in G$. As $G$ is given by its multiplication table we may, in $\ACz$, compute the prime divisors of $|Z(G)|$ \cite[Lemma~2.3]{CGLWISSAC}. If more than one prime divides $|Z(G)|$, we reject. So now, let $p$ be the unique prime divisor of $|Z(G)|$. Now $Z(G)$ is elementary Abelian if and only if for each $1 \neq x \in Z(G)$, $|x| = p$. We can verify this condition in $\textsf{L}$ \cite{BKLM, CGLWISSAC}. Next, note that $Z(G) = \text{Rad}(G)$ if and only if $G/\text{Rad}(G)$ has no Abelian normal subgroups. We may write down the multiplication table for $G/Z(G)$ in $\ACz$. We can then decide if $G/Z(G)$ has no Abelian normal subgroups in $\ACz$ using \cite[Lemma~7.11 in arXiv Version]{GrochowJohnsonLevet}. Thus, at this stage, we may now assume that $Z(G) = \text{Rad}(G)$. 
\begin{enumerate}[label=(\alph*)]
\item We now turn to checking whether $G/Z(G)$ decomposes as a direct product of non-Abelian simple groups. We may compute such a decomposition or decide that no such decomposition exists, in $\mathsf{FL}$ \cite[Lemma~6.11]{GrochowLevetWL}. 

\item We now turn to checking whether $G/Z(G)$ decomposes as a direct product of $O(1)$-size perfect groups. As $G/Z(G)$ has at most $\log |G/Z(G)| \leq \log n$ direct factors, we may in $\ACz$, enumerate the possible direct factors (by listing the elements of each proposed direct factor) and test whether each purported direct factor satisfies the group axioms. As each direct factor has bounded size, we may in $\ACz$ test whether the subgroups $T_1, \ldots, T_{\ell}$ we have enumerated are perfect, normal in $G/Z(G)$, and satisfy $T_{i} \cap T_{j} = \emptyset$ for all distinct $i, j$. Furthermore, we may test in $\textsf{L}$ whether $\prod_{i=1}^{\ell} |T_{i}| = |G/Z(G)|$, which ensures that $\prod_{i=1}^{\ell} T_{i} = G/Z(G)$. The total work in this case is computable in $\textsf{FL}$.
\end{enumerate}

Now given the direct factors, we may use the generator-enumeration strategy to decide when two such factors are isomorphic in $\mathsf{FL}$ \cite{TangThesis}. Thus, we may group the direct factors by isomorphism type in $\mathsf{FL}$. The result now follows.
\end{proof}

\begin{definition} \label{def:Diagonal}
Let $Q = \prod_{i=1}^{\ell} T_{i}$. Classify the $T_{i}$'s together and group them by their isomorphism types, identifying $Q = \prod_{i=1}^{r} Q_{i}^{\ell_{i}}$. Then $\Aut(Q) \cong \prod_{i=1}^{r} \Aut(Q_{i}) \wr \text{Sym}(\ell_{i})$. A \emph{diagonal} of $\Aut(Q)$ is an element in $\prod_{i=1}^{r} \Aut(Q_{i})$.
\end{definition}

\begin{lemma} \label{lem:EnumerateDiagonals}
Let $Q$ be as defined in Definition~\ref{def:Diagonal}, and suppose that each $T_{i}$ is $O(1)$-generated. Suppose that we are given the decomposition $Q = \prod_{i=1}^{r} T_{i}^{\ell}$. We can enumerate all diagonals of $Q$ in $\mathsf{FL}$.
\end{lemma}

\begin{proof}
Grochow and Qiao \cite[Proof of Theorem~7.1]{GQCoho} established that we can enumerate all diagonals of $Q$ using the generator-enumeration technique, which is $\mathsf{FL}$-computable \cite{TangThesis}.
\end{proof}

\begin{lemma}[{\cite[Lemma~7.5]{GQCoho}}] \label{lem:7.5}
Let $A' \times A''  \xhookrightarrow{\iota} G \overset{\pi}{\twoheadrightarrow} Q$ be a central extension of $A' \times A''$ by $Q$. Let $p_{A'} : A' \times A'' \to A'$ be the projection onto $A'$ along $A''$. If there is a $2$-cocycle $f : Q \times Q \to A' \times A''$ such that $p_{A'} \circ f$ is a $2$-coboundary, then $G$ is isomorphic (even equivalent) to $A' \times (G/A')$. 

Furthermore, if $Z(G)$ is elementary Abelian, then $A'$ can be computed in $\textsf{NC}^{2}$ using linear algebra over Abelian groups.
\end{lemma}

\begin{proof}[Proof Sketch.]
Grochow and Qiao only claimed polynomial-time bounds for computing $A'$. The $\textsf{NC}^{2}$ bound follows from using Proposition~\ref{prop:GQ6.8} in place of \cite[Proposition~6.4]{GQCoho} and known $\textsf{NC}^{2}$ algorithms to solve systems of linear equations over $\mathbb{Z}_{p}$ (see e.g., \cite{BorodinParallelMatrix, MulmuleyRank, BDHMParallelLinearAlgebra}).
\end{proof}

\begin{definition}[{\cite[Proof of Theorem~7.1]{GQCoho}}]
Let $G$ be a group under consideration in Theorem~\ref{thm:MainCohoSec7}. We say that a $2$-cocycle $f_{j} : Q \times Q \to A$ respects the direct factors if there exist $f_{j,i} : T_{i} \times T_{i} \to A$ ($i \in [\ell]$) such that the following condition holds:
\[
f_{j}((p_{1}, \ldots, p_{\ell}), (q_{1}, \ldots, q_{\ell})) = \sum_{i \in [\ell]} f_{j,i}(p_{i}, q_{i}).
\]
Let $Z_{\text{prod}}^{2}(Q,A)$ denote the set of $2$-cocycles respecting the direct factors. Grochow and Qiao \cite[Lemma~7.4]{GQCoho} showed that $Z_{\text{prod}}^{2}(Q,A) \neq \emptyset$. Similarly, we say that a $2$-coboundary $b_{j} : Q \times Q \to A$ respects the direct factors if there exist $b_{j,i} : T_{i} \times T_{i} \to A$ ($i \in [\ell]$) such that the following condition holds:
\[
b_{j}((p_{1}, \ldots, p_{\ell}), (q_{1}, \ldots, q_{\ell})) = \sum_{i \in [\ell]} b_{j,i}(p_{i}, q_{i}).
\]
Let $B_{\text{prod}}^{2}(Q,A)$ be the set of $2$-coboundaries respecting the direct factors. The difference of two cohomologous $2$-coycles in $Z_{\text{prod}}^{2}(Q,A)$ belongs to $B_{\text{prod}}^{2}(Q,A)$. Define $C_{\text{prod}}^{2}(Q,A)$ as the set of $2$-cochains respecting the direct factors. We may view the elements of $C_{\text{prod}}^{2}(Q,A)$ as $k \times \left( \sum_{i \in [\ell]} |T_{i}|^{2} \right)$ matrices, whose rows are indexed by $[k]$ and whose columns are indexed by triples $(i; p, q)$ with $p, q \in T_{i}$. That is, $C_{\text{prod}}^{2}(Q,A) = \bigoplus_{i \in [k]} C^{2}(T_{i}, A)$.
\end{definition}

%%%%
\subsection{Proof of Theorem~\ref{thm:MainCohoSec7}}

We now turn to proving Theorem~\ref{thm:MainCohoSec7}. 

\begin{proof}[Proof of Theorem~\ref{thm:MainCohoSec7}]
We follow the strategy of Grochow and Qiao \cite[Section~7.3]{GQCoho}. We begin by decomposing $G_{j}$ ($j = 1, 2$) as an extension of $A = \rad(G_{i}) = \mathbb{Z}_{p}^{k}$ by $Q = \prod_{i=1}^{\ell} T_{i}$. Fix arbitrary sections $s_{j} : Q \to G_{j}$ ($j = 1, 2$), and let $f_{j}$ be the corresponding $2$-cocycles. We then classify the $T_{i}$'s according to their isomorphism types. Grouping them together by isomorphism type, we have $Q = \prod_{i=1}^{r} Q_{i}^{\ell_{i}}$. By Lemma~\ref{lem:DecomposeCentralRadical}, this step is $\textsf{FL}$-computable. We may now enumerate all diagonals of $\Aut(Q)$ in $\textsf{FL}$ using Lemma~\ref{lem:EnumerateDiagonals}. By Lemma~\ref{lem:GQMainLemma}, $G_{1} \cong G_{2}$ if and only if they are pseudo-congruent extensions of $A$ by $Q$. The extensions are pseudo-congruent if and only if there exists $(\alpha, \beta) \in \Aut(A) \times \Aut(Q)$ such that, after twisting by $(\alpha, \beta)$, the resulting extensions are equivalent. Once we fix such an $(\alpha, \beta) \in \Aut(A) \times \Aut(Q)$, Lemma~\ref{lem:7.4} provides that the problem is reduced to determining the equivalence of $G_{1}|_{T_{i}}$ and $G_{2}|_{T_{i}}$ ($i \in [\ell]$).
\begin{itemize}
\item \textbf{Deciding Isomorphism.} We first show how to decide isomorphism.
\begin{enumerate}[label=(\alph*)]
\item Consider first the case when each $T_{i}$ is non-Abelian simple. In this case, each $T_{i}$ is $2$-generated, we can list $\Aut(T_{i})$ in $\textsf{FL}$ \cite{TangThesis}. By Theorem~\ref{thm:GQCohoA}, we can determine the equivalence type of $G|_{T_{i}}$ in $\textsf{AC}^{2}$. 

Note that every $\beta \in \Aut(Q)$ can be represented as a pair ($\delta, \sigma) \in (\prod_{i=1}^{r}  \Aut(Q_{i})^{\ell_{i}}) \rtimes (\prod_{i=1}^{r} \text{Sym}(\ell_{i}))$. Thus, $Z_{\text{prod}}^{2}(Q,A)$ is an invariant subset of $Z^{2}(Q,A)$ under the actions of both $\Aut(A)$ and $\Aut(Q)$. Similarly, $B_{\text{prod}}^{2}(Q,A)$ is an invariant subset of $B^{2}(Q,A)$ under the actions of both $\Aut(A)$ and $\Aut(Q)$. By Proposition~\ref{prop:6.10}, we can in $\textsf{NC}^{2}$, compute for each $i \in [\ell]$ an $\Aut(A)$-invariant $W_{i}$ such that $C^{2}(T_{i},A) =   B^{2}(T_{i}, A) \oplus W_{i}$, along with a projection $\pi_{i} : C^{2}(T_{i}, A) \to W_{i}$ such that $\pi_{i}$ commutes with the action of $\Aut(A)$. 

Grochow and Qiao showed that each element of $W_{i}$ ($i \in [\ell]$) can be written as a $k \times 17$ matrix, in a manner that is still $\Aut(A)$-equivariant. For the remainder of this proof, we will denote $\pi_{i}$ as the composition of the previous $\pi_{i}$, followed by the mapping onto $k \times 17$ matrices. 

Furthermore, Grochow and Qiao established that for each choice of diagonal $\delta = \prod_{i=1}^{\ell} \Aut(T_{i})$, we may choose the complements $W_{i}, W_{h}$ such that whenever $T_{i} \cong T_{h}$, $\delta$ identifies $W_{i}$ and $W_{h}$. With this choice, we can construct 
\[
\pi_{\delta} = \bigoplus_{i=1}^{\ell} \pi_{i} : C^{2}(Q, A) \to W \leq \text{Mat}_{k \times 17\ell}(\mathbb{Z}_{p}),
\]
for some $\Aut(A)$-invariant $W \leq C_{\text{prod}}^{2}(Q,A)$ such that $C_{\text{prod}}^{2}(Q,A) = B_{\text{prod}}^{2}(Q,A) \oplus W$ and $\pi_{\delta}$ commutes with the action of $\Aut(A) \times \prod_{i=1}^{r} \text{Sym}(\ell_{i})$. For each diagonal $\delta$, it remains to decide whether there exists $(\alpha, \sigma) \in \Aut(A) \times \prod_{i=1}^{r} \text{Sym}(\ell_{i})$ such that $\pi(f_{1}) = \pi(f_{2}^{(\text{id}, \delta, \text{id})})^{(\alpha, \text{id}, \sigma)}$. Note that as $\alpha, \delta$ commute and $\sigma$ is already on the right, we have that $(\alpha, \delta, \sigma) = (\text{id}, \delta, \text{id})(\alpha, \text{id}, \sigma)$. 

Let $M_{1} := M_{\pi(f_{1})}$. We may assume that $M_{1}$ has rank $k$. Otherwise, we may in $\textsf{NC}^{2}$ (using Lemma~\ref{lem:7.5}), split out a direct factor of the center as $A' \times G_{j}/A'$ ($j = 1, 2$). Note that Lemma~\ref{lem:7.5} explicitly returns such an $A'$ as a direct factor of $Z(G_{j})$ ($j = 1, 2$). Grochow and Qiao note that while Lemma~\ref{lem:7.5} is concerned with $Z^{2}(Q,A)$ and $B^{2}(Q,A)$, it is readily adapted to the setting of $Z_{\text{prod}}^{2}(Q,A)$ and $B_{\text{prod}}^{2}(Q,A)$. By the Remak--Krull--Schmidt theorem, we then reduce to testing isomorphism between $G_{1}/A'$ and $G_{2}/A'$, where the desired rank condition holds.  

Now in parallel, for each diagonal $\delta$ of $\prod_{i=1}^{\ell} T_{i}$, we compute $\pi(f_{1})$ and $\pi(f_{2}^{(\text{id}, \delta, \text{id})})$. Let $M_{1} = M_{\pi(f_{1})}$ and $M_{2} = M_{\pi(f_{2}^{(\text{id}, \delta, \text{id})})}$. Again note that we can enumerate all such diagonals in $\textsf{FL}$ (Lemma~\ref{lem:EnumerateDiagonals}). Furthermore, we have established above that $\pi$ can be constructed in $\textsf{NC}^{2}$.

As the action of $\Aut(A)$ on $M_{1}, M_{2}$ is by left multiplication, and the action of $\prod_{i=1}^{r} \text{Sym}(\ell_{i})$ is on the blocks of columns, we treat $M_{1}, M_{2}$ as generators of two $\mathbb{Z}_{p}$-codes of dimension $k$ and length $17\ell$. As $\ell \leq \log n$, we compute the coset of equivalences $\text{CodeEq}(M_{1}, M_{2}) \subseteq \text{Sym}(17\ell)$ using an $\textsf{AC}$ circuit of depth $O((\log^{2} n)(\poly(\log \log n)))$ and size $\poly(n)$ (Theorem~\ref{thm:LinearCodeEquivalence}). We then compute $\text{CodeEq}(M_{1}, M_{2}) \cap \prod_{i=1}^{r} \text{Sym}(\ell_{i})$ in \FOPLL using Lemma~\ref{PermutationGroupsNC}\ref{CosetInt}. If this intersection is non-empty, then we report isomorphic. Using a single $\textsf{OR}$ gate, we return true (report isomorphic) if and only if at least one of these intersections is non-empty (where the \textsf{OR} is taken over all diagonals of $Q$). 

The total work is computable using an $\textsf{AC}$ circuit of dpeth $O((\log^{2} n)(\poly(\log \log n)))$ and size $\poly(n)$. The $\textsf{AC}^{3}$ bound now follows in this case.

\item Consider now the case in which each $T_{i}$ has size at most $D \in O(1)$. This case is handled almost identically as part (a). We will discuss the differences here. Note that in this case, we can determine the isomorphism type of a given $T_{i}$ in $\ACz$, by trying all of the at most $D! \in O(1)$ permutations. Grochow and Qiao establishd that the corresponding linear codes have length $\ell \cdot D^2$, rather than $17\ell$. As $\ell \leq \log n$, we compute the coset of equivalences $\text{CodeEq}(M_{1}, M_{2})$ using an $\textsf{AC}$ circuit of depth $O((\log^{2} n)(\poly(\log \log n)))$ and size $\poly(n)$ (Theorem~\ref{thm:LinearCodeEquivalence}). We then compute $\text{CodeEq}(M_{1}, M_{2}) \cap \prod_{i=1}^{\ell} \text{Sym}(\ell_{i})$ in \FOPLL using Lemma~\ref{PermutationGroupsNC}\ref{CosetInt}. If this intersection is non-empty, then we report isomorphic. Using a single $\textsf{OR}$ gate, we return true (report isomorphic) if and only if at least one of these intersections is non-empty. 

The total work is computable using an $\textsf{AC}$ circuit of dpeth $O((\log^{2} n)(\poly(\log \log n)))$ and size $\poly(n)$. The $\textsf{AC}^{3}$ bound now follows in this case.
\end{enumerate}

\item \textbf{Computing the coset of isomorphisms.} The algorithm is essentially the same for both cases (a) and (b). This is also identical to that of Grochow and Qiao \cite[Section~7.3]{GQCoho}. Here, we use our parallelization of \algprobm{Linear Code Equivalence} (Theorem~\ref{thm:LinearCodeEquivalence}), the \FOPLL implementation for \emph{small} instances of \algprobm{Coset Intersection} (Lemma~\ref{PermutationGroupsNC}\ref{CosetInt}), and parallel linear algebra (see e.g., \cite{BorodinParallelMatrix, MulmuleyRank, BDHMParallelLinearAlgebra}) to obtain improvements. We now proceed to parallelize the work of Grochow and Qiao.

Let $\delta$ be a diagonal such that $\text{CodeEq}(M_{1}, M_{2}) \cap \prod_{i=1}^{r} \text{Sym}(\ell_{i})$ is non-empty. By the above work on deciding isomorphism, we can find such a $\delta$ in $\textsf{AC}^{3}$. We pick a single $\sigma \in \text{CodeEq}(M_{1}, M_{2}) \cap \prod_{i=1}^{r} \text{Sym}(\ell_{i})$. Now we apply $(\delta, \sigma)$ to the columns of $M_{2}$, to obtain a matrix $M_{2}'$. Next, we use parallel linear algebra (see e.g., \cite{BorodinParallelMatrix, MulmuleyRank, BDHMParallelLinearAlgebra}) to, in $\textsf{NC}^{2}$, find a matrix $X$ in $\Aut(A) \cong \text{GL}_{k}(\mathbb{Z}_{p})$ $XM_{1} = M_{2}'$.

We now turn to constructing a generating set for $\Aut(G_1)$. As in the proof of Theorem~\ref{thm:GQCohoA}(b), we can compute the homomorphisms $Q \to A$ in $\textsf{NC}^{2}$ using parallel linear algebra. These constitute the kernel of the map $\Aut(G_1) \to \Aut(A) \times \Aut(Q)$. We now apply Observation~\ref{obs:5.2}, viewing $\Aut(G_1)$ as an extension of $\Aut(A) \times \prod_{i=1}^{r} \Aut(Q_i)^{\ell_{i}}$ by $\prod_{i=1}^{r} \text{Sym}(\ell_{i})$. For each diagonal $\delta \in \prod_{i=1}^{r} \Aut(Q_i)^{\ell_{i}}$, we compute $\text{CodeEq}(M_1, M_1) \cap \prod_{i=1}^{r} \text{Sym}(\ell_{i})$. We can compute $\text{CodeEq}(M_1, M_1)$ using an $\textsf{AC}$ circuit of depth $O((\log^2 n)(\poly(\log \log n)))$ and size $\poly(n)$ (Theorem~\ref{thm:LinearCodeEquivalence}), followed by the intersection in \FOPLL (Lemma~\ref{PermutationGroupsNC}\ref{CosetInt}). The union of these coset intersections generates the image of $\Aut(G_{1})$ in $\prod_{i=1}^{r} \text{Sym}(\ell_i)$. It remains to compute generators of the kernel of the map $\Aut(G_1) \to \prod_{i=1}^{r} \text{Sym}(\ell_i)$, as a subgroup of $\Aut(A) \times \prod_{i=1}^{r} \Aut(Q_i)^{\ell_i}$. 

We will now apply Observation~\ref{obs:5.2} again. Let $\Aut_{0}(G_1)$ denote the subgroup of $\Aut(G_1)$ whose projection onto $\prod_{i=1}^{r} \text{Sym}(\ell_i)$ is trivial. We can determine the projection of $\Aut_{0}(G_1)$ to $\prod_{i=1}^{r} \Aut(Q_{i})^{\ell_i}$ by enumerating over all diagonals, and including only those diagonals $\delta$ for which the matrix associated to $\pi(f_1)$ (which we call $M_{\pi(f_1)}$) and the matrix associated to $\pi(f_{1}^{(\text{id}, \delta, \text{id})})$ (which we call $M_{\pi(f_{1}^{(\text{id}, \delta, \text{id})})}$) have the same rowspan. We may enumerate all such diagonals in $\textsf{FL}$ using Lemma~\ref{lem:EnumerateDiagonals}, and then compute the rowspans in $\textsf{NC}^{2}$ (see e.g., \cite{BorodinParallelMatrix, MulmuleyRank, BDHMParallelLinearAlgebra}). For each such diagonal $\delta$, we obtain a corresponding lift to $\Aut_{0}(G_1)$ by using linear algebra to, in $\textsf{NC}^{2}$, find a matrix $X$ of $\Aut(A) \cong \text{GL}_{k}(\mathbb{Z}_{p})$ such that $XM_{\pi(f_1)} = M_{\pi(f_{1}^{(\text{id}, \delta, \text{id})})}$. 

We finally need to compute generators of $\Aut(A)$ that fixes $\pi(f_1)$. It was established in the proof of Theorem~\ref{thm:GQCohoA}(b) that this step is $\textsf{NC}^{2}$-computable.

The total work is computable in $\textsf{AC}^{3}$, as desired. The result now follows. \qedhere
\end{itemize}
\end{proof}

%%%%%%%
\section{Conclusion}

We exhibited an $\textsf{AC}^{3}$ isomorphism test for the class $\mathcal{H}(\mathcal{A}, \mathcal{E})$ of coprime extensions $H \ltimes N$, where $N$ is Abelian and $H$ is elementary Abelian. Additionally, we showed that in $\textsf{AC}^{3}$, we can decide isomorphism and even compute the coset of isomorphisms between two central, elementary-Abelian radical groups, where $G/\rad(G)$ is a direct product of either non-Abelian simple groups or $O(1)$-size perfect groups. In the process, we showed that isomorphism testing of arbirary central-radical groups is decidable using $\textsf{AC}$ circuits of depth $O(\log^3 n)$ and size $n^{O(\log \log n)}$. Our work leaves open several questions.

\begin{question}
Can the bounds for either Theorem~\ref{thm:MainCoprime} or Theorem~\ref{thm:MainCoho} be improved to $\textsf{L}$?
\end{question}

Isomorphism testing of many natural families of graphs, including graphs of bounded treewidth \cite{ElberfeldSchweitzer} and graphs of bounded genus \cite{ElberfeldKawarabayashiGenus}, are $\textsf{L}$-complete. While \algprobm{GpI} is not $\textsf{L}$-complete under $\textsf{AC}^{0}$-reductions \cite{ChattopadhyayToranWagner, CGLWISSAC}, establishing an upper bound of $\textsf{L}$ for isomorphism testing nonetheless remains a natural target. 

Kantor, Luks, and Mark \cite{KantorLuksMark} exhibited an $\textsf{NC}$ algorithm for \algprobm{Transversal} when $|G| \in m^{\mathrm{polylog}(m)}$. They claimed to have a proof, to appear in future work, that \algprobm{Transversal} was $\textsf{NC}$-computable without imposing restrictions on $|G|$. However, we were unable to find this stronger result in the literature. We thus ask the following.

\begin{question} \label{q:Transversal}
Show that \algprobm{Transversal} belongs to $\textsf{NC}$ for arbitrary $G$.
\end{question}

Resolving Question~\ref{q:Transversal} would yield an \FOPLL bound when our permutation domain has size $O(\log n)$. Consequently, this would decrease the circuit depth for our \algprobm{Linear Code Equivalence} procedure from $O((\log^2 n)(\poly(\log \log n)))$ to $O((\log n)(\poly(\log \log n)))$.

Another bottleneck for improving the circuit depth for isomorphism testing of $\mathcal{H}(\mathcal{A}, \mathcal{E})$ is in factoring polynomials over finite fields. This problem is solvable in $\textsf{NC}^{3}$ \cite{Berlekamp} (see \cite{Eberly}). We ask the following.

\begin{question}
Is it possible to factor polynomials over finite fields in $\textsf{TC}^{2}$? 
\end{question}

\section*{Acknowledgements}

I wish to thank Peter Brooksbank, Edinah Gnang, Joshua A. Grochow, Takunari Miyazaki, and James B. Wilson for helpful discussions, as well as the anonymous referees for helpful feedback. Parts of this work began at Tensors Algebra-Geometry-Applications (TAGA) 2024. I wish to thank Elina Robeva, Christopher Voll, and James B. Wilson for organizing this conference. I was partially supported by CRC 358 Integral Structures in Geometry and Number Theory at Bielefeld and Paderborn, Germany; the Department of Mathematics at Colorado State University; James B. Wilson's NSF grant DMS-2319370; and travel support from the Department of Computer Science at the College of Charleston.

\bibliographystyle{alphaurl}
\bibliography{references}

\end{document}